\documentclass[twocolumn,aps,pra,superscriptaddress,nofootinbib,10pt]{revtex4-2}

\usepackage{amsthm}
\usepackage{amsmath}
\usepackage{amssymb}
\usepackage{thmtools}
\usepackage{thmtools, thm-restate}
\declaretheorem[name=Theorem]{thm}

\declaretheorem[name=Lemma]{lemma}

\usepackage[dvipsnames,table]{xcolor}
\usepackage{graphicx}

\usepackage[justification=Justified]{caption}
\usepackage{subcaption}
\usepackage{mathtools}

\usepackage{float}
\usepackage{tabularx}
\usepackage{multirow}
\usepackage{braket}
\usepackage{bm}

\usepackage[unicode=true]{hyperref}

\usepackage[algo2e,ruled]{algorithm2e}
\SetKwComment{Comment}{/* }{ */}

\begin{document}

	\title{Fast quantum integer multiplication with zero ancillas}
	\author{Gregory D. Kahanamoku-Meyer}
	\email{Corresponding author. Email: gkm@berkeley.edu}
	\affiliation{Department of Physics, University of California at Berkeley, Berkeley, CA 94720}
	\affiliation{Research Laboratory of Electronics, Massachusetts Institute of Technology, Cambridge, MA 02139}
	\author{Norman Y. Yao}
	\affiliation{Department of Physics, University of California at Berkeley, Berkeley, CA 94720}
	\affiliation{Department of Physics, Harvard University, Cambridge, MA 02139}

	\begin{abstract}
          The multiplication of superpositions of numbers is a core operation in many quantum algorithms.
          The standard method for multiplication (both classical and quantum) has a runtime quadratic in the size of the inputs.
          Quantum circuits with asymptotically fewer gates have been developed, but generally exhibit large overheads, especially in the number of ancilla qubits.
          In this work, we introduce a new paradigm for sub-quadratic-time quantum multiplication with zero ancilla qubits---the only qubits involved are the input and output registers themselves.
          Our algorithm achieves an asymptotic gate count of $\mathcal{O}(n^{1+\epsilon})$ for any $\epsilon > 0$; with practical choices of parameters, we expect scalings as low as $\mathcal{O}(n^{1.3})$.
          Used as a subroutine in Shor's algorithm, our technique immediately yields a factoring circuit with $\mathcal{O}(n^{2+\epsilon})$ gates and only $2n + \mathcal{O}(\log n)$ qubits; to our knowledge, this is by far the best qubit count of any factoring circuit with a sub-cubic number of gates.
          Used in Regev's recent factoring algorithm, the gate count is $\mathcal{O}(n^{1.5+\epsilon})$.
          Finally, we demonstrate that our algorithm has the potential to outperform previous proposals at problem sizes relevant in practice, including yielding the smallest circuits we know of for classically-verifiable quantum advantage.
	\end{abstract}

	\maketitle

    \section{Introduction}
    
    Quantum circuits that perform arithmetic on superpositions of numbers are the fundamental building block of many quantum algorithms, from factoring to protocols for certifiable random number generation and efficiently-verifiable quantum computational advantage~\cite{shor_polynomial-time_1997, mahadev_classical_2018, gheorghiu_computationally-secure_2019, brakerski_cryptographic_2021, brakerski_simpler_2020, kahanamoku-meyer_classically_2022, brakerski_simple_2023, regev_efficient_2023}. 
    The feasibility of realizing these applications depends on the efficiency with which this arithmetic, and in particular multiplication, can be implemented.
    The standard way of performing multiplication, in both the classical and quantum setting, is via the ``schoolbook'' algorithm using $\mathcal{O}(n^2)$ gates, where $n$ is the size of the input.
    Asymptotically sub-quadratic-time algorithms have been known for over half a century, but have overheads that make them useful only for the multiplication of very large values.\footnote{For a classical example, the GNU multiple-precision arithmetic library uses a threshold of 2176 bit inputs to switch away from the schoolbook method.}
    In the quantum setting,  the reversibility constraint imposed by  unitarity generally makes these overheads even worse. 
         
    Despite this challenge, several works have explored  the implementation of quantum circuits for sub-quadratic-time multiplication, largely focusing on the Karatsuba algorithm which exhibits an asymptotic run-time of roughly $\mathcal{O}(n^{1.58})$~\cite{parent_improved_2018, dutta_quantum_2018,
        gidney_asymptotically_2019, larasati_quantum_2021, kowada_reversible_2006}.
	A significant obstacle to this effort is the recursive structure of fast multiplication algorithms---making them reversible involves storing intermediate data, which ultimately requires utilizing a large number of ancilla qubits. 
    Notable recent work has, for the first time, reduced the number of ancillas required for quantum Karatsuba multiplication to linear  in the size of the inputs~\cite{gidney_asymptotically_2019}.
    However, in practice, the number of ancillas still dominates the total qubit cost. 
    Moreover, it has remained an open question whether sub-quadratic-time quantum multiplication with fewer than $\mathcal{O}(n)$ ancillas is even possible.  
    
    \begin{figure*}
    	\includegraphics[width=\textwidth]{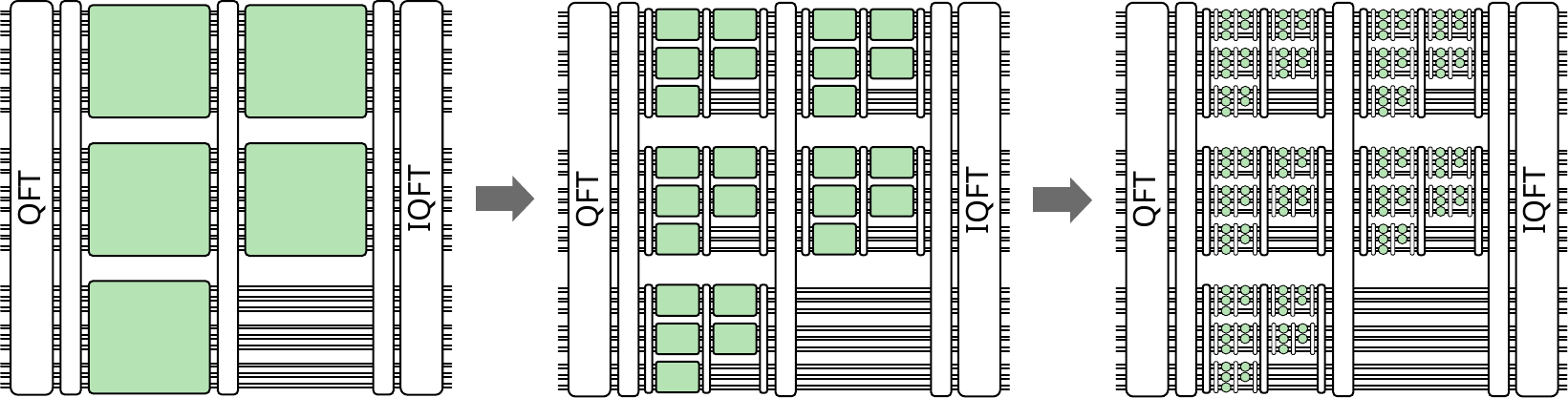}
    	\caption{\textbf{Recursive structure of the quantum multiplication circuit, for $k=3$ classical-quantum multiplication.} The green boxes represent phase rotations proportional to the product of registers (the \textsf{PhaseProduct} operation described in Section~\ref{sec:algorithm}), which are recursively decomposed; the unlabeled tall, narrow unitaries represented by white boxes are in-place quantum addition circuits. The qubits of the input and output registers are interleaved; the top pair of qubits corresponds to the least significant bit of the input and output registers, the next pair the second least significant bits of each, and so on. At the last level of recursion, phase rotations on individual pairs of qubits can be implemented directly as two-qubit controlled phase rotation gates (represented here by green circles). Note that the locality of the gates improves as the recursion proceeds.}
    	\label{fig:circuit}
    \end{figure*}

    In this work, we introduce a new paradigm for sub-quadratic-time multiplication on a quantum computer, that requires zero ancilla qubits (Fig.~\ref{fig:circuit}): the only qubits needed are those used to store the input and output data.
  	In terms of gate counts, our algorithms achieve asymptotic scaling as low as $\mathcal{O}(n^{1+\epsilon})$ for any $\epsilon > 0$ (see Table~\ref{tab:scaling}).
    Our approach combines ideas from classical fast multiplication  with an inherently quantum technique where arithmetic is performed in the phase of a quantum state. 
	Using our algorithm as a subroutine in Shor's algorithm for factoring yields circuits that require only $2n + \mathcal{O}(\log n)$ total qubits, where $n$ is the length of the integer to be factored, while obtaining an asymptotic gate count of $\mathcal{O}(n^{2+\epsilon})$ for arbitrarily small $\epsilon$; to the best of our knowledge, this represents by far the fewest qubits required to perform Shor's algorithm in fewer than $\mathcal{O}(n^{3})$ gates.\footnote{See Table~1 of~\cite{gidney_how_2021} for a recent comparison of proposals for the implementation of Shor's algorithm.}
	Moreover, used as a subroutine in a space-optimized version of Regev's recent factoring algorithm, our multiplication algorithm yields quantum circuits for factoring with as few as $\mathcal{O}(n^{1.5 + \epsilon})$ gates and $\mathcal{O}(n)$ qubits~\cite{regev_efficient_2023, ragavan_space-efficient_2024}.

    One might naturally wonder whether our algorithm's elimination of ancilla qubits comes at the cost of significant overhead in constant factors.
    This is not the case.
    Indeed, for multiplication of a quantum integer register by a classical value (the key operation for implementing Shor's algorithm), we find it is possible to construct circuits for which the number of qubits required is dramatically reduced, while the gate count remains small~(see Table~\ref{tab:cq-2048-costs}).
    Intuitively, this gate efficiency comes from a fundamental feature of our algorithm's structure: that a substantial portion of the overhead implicit in fast multiplication can be performed in classical precomputation. 
    Highlighting the versatility of our algorithm, we also investigate its application to an efficiently-verifiable proof of quantumness~\cite{kahanamoku-meyer_classically_2022}~(see Table~\ref{tab:x2modN-costs}).
    
    \section{Background and Framework}
    
    We focus on the implementation of two related unitaries, which are  defined by their action on product states (and can be extended by linearity to superpositions of inputs):
    \begin{align}
    	\mathcal{U}_{c \times q}(a) \ket{x} \ket{w} &= \ket{x} \ket{w + ax} \label{eq:cq_unitary}\\
    	\mathcal{U}_{q \times q} \ket{x} \ket{y} \ket{w} &= \ket{x} \ket{y} \ket{w + xy}. \label{eq:qq_unitary}
    \end{align}
 The \emph{classical-quantum multiply}, $\mathcal{U}_{c \times q} (a)$, corresponds to the multiplication of an integer stored in a quantum register,  $\ket{x}$, by a classical number, $a$.
 The result is added to an output register that initially contains the value $\ket{w}$.
Similarly, the \emph{quantum-quantum multiply}, $\mathcal{U}_{q \times q}$, corresponds to the multiplication of a quantum integer  $\ket{x}$ by another quantum integer $\ket{y}$.
The value is again added to an output register $\ket{w}$.

In order to perform both of the above unitaries in sub-quadratic-time without ancillas, we draw inspiration from two key ideas: (i) the classical fast (sub-quadratic) multiplication algorithm of Toom and Cook and (ii) the quantum Fourier transform based multiplication algorithm of Draper.
At first glance, these two seem entirely incompatible.
Toom-Cook gains its speed from the reuse of intermediate values that have already been computed; however, Draper arithmetic is performed via quantum phase rotations, which, once applied, cannot be reused. 
The crux of our result uses an inexpensive classical precomputation to obviate the reuse of intermediate values; this simultaneously removes the need to store intermediate values, while also reducing the size of the quantum circuit.

	\subsection{Classical Toom-Cook multiplication}
	\label{subsec:toom-cook}

	The most straightforward algorithm for multiplication is known as the \textit{schoolbook} algorithm and has a runtime of $\mathcal{O}(n^2)$.
 It decomposes the product as $xy = \sum_{i,j} b^{i+j} x_i y_j$ for a base $b$, where $x_i$ is the $i^\mathrm{th}$ base-$b$ digit of $x$ (and $y_j$ is defined similarly).
	With $b=2$, the circuit is quite simple, because the binary product $x_i y_j$ is an AND gate.
	
	Starting in the 1960s, it was realized that  multiplication could be performed faster than $\mathcal{O}(n^2)$ by breaking the inputs into pieces, and carefully combining linear combinations of those pieces.
To establish the notation that we will use throughout this work, we begin by expressing the classical Toom-Cook multiplication algorithm in the language of linear algebra~\cite{bodrato_towards_2007, knuth_art_1998}.
	
Consider an $n$-bit value $x$ written in base  $b = 2^{n/k}$ for some integer $k$, such that $
		x = \sum_{i=0}^{k-1} x_i b^i$.
In this representation, each of the $k$ digits, $x_i$, corresponds to $n/k$ bits of the binary representation of $x$.
The above sum can be expressed as the inner product of two length-$k$ vectors: $\mathbf{x} = (x_{k-1}, \cdots, x_1, x_0)$ and $\mathbf{e}_b = (b^{k-1}, \cdots, b, 1)$,
\begin{equation}
		x = \sum_{i=0}^{k-1} x_i b^i = \mathbf{x}^{\intercal}  \mathbf{e}_b.
		\label{eq:toom-cook-poly}
	\end{equation}

	Allowing $b$ to vary, Eq.~\ref{eq:toom-cook-poly} defines a polynomial $x(b)$ whose coefficients are the elements of $\mathbf{x}$.
	Letting $y(b)$ be the polynomial corresponding to another integer $y$, one can define a third polynomial $p(b) = x(b) y(b)$ for which $p(2^{n/k})$ equals the numerical value of the integer product $p = xy$.
	Thus, the problem of integer multiplication can be expressed as polynomial multiplication.  
	
	The key insight of the Toom-Cook algorithm is that the polynomial $p(b)$ is uniquely determined by its value at $q=2k-1$ points $w_\ell$, where $\ell \in \{0, \cdots,q-1\}$, and that each value, $p(w_\ell)$, can be computed as the pointwise multiplication $x(w_\ell) y(w_\ell)$.
	With an appropriate choice of $w_\ell$, the values $x(w_\ell)$ and $y(w_\ell)$ are integers of roughly $n/k$ bits. Thus $p(b)$ can be computed via only $q=2k-1$ multiplications (compare to the $k^2$ required by the schoolbook algorithm) of size $\sim n/k$.
	By recursively applying this construction to the smaller products of size $n/k$, and then the resulting products  of size $n/k^2$ (and so on),  the Toom-Cook algorithm achieves an asymptotic runtime of $\mathcal{O}(n^{\log_k q})$, which is sub-quadratic for all $k \geq 2$.

 The evaluation of the polynomial $x(b)$ at the points $w_\ell$ can be expressed via the matrix $A$ whose $\ell^\mathrm{th}$ row is $\mathbf{e}_{w_\ell}$ (which is known as a Vandermonde matrix):
	\begin{equation}
		A = \begin{pmatrix}
			w_0^{q-1} & \cdots & w_0^2 & w_0 & 1 \\
			w_1^{q-1} & \cdots & w_1^2 & w_1 & 1 \\
			\vdots    & \ddots & \vdots & \vdots & \vdots \\
			w_{q-1}^{q-1} & \cdots & w_{q-1}^2 & w_{q-1} & 1 \\
		\end{pmatrix}.
	\end{equation}
	Then, the vector $\tilde{\mathbf{x}}$ containing the values $x(w_\ell)$ is given by $\tilde{\mathbf{x}} = A\mathbf{x}$, and the vector $\tilde{\mathbf{p}}$, containing the values $p(w_\ell)$, can be expressed as
		$\tilde{\mathbf{p}} = \tilde{\mathbf{x}} \circ \tilde{\mathbf{y}} = A\mathbf{x} \circ A\mathbf{y}$,
 where  $\circ$ denotes the pointwise  product of vectors.
	(We implicitly extend the vectors $\mathbf{x}$ and $\mathbf{y}$ to length $q$ elements by inserting zeros for all indices larger than $k-1$).
 
	As long as the points $w_\ell$ are chosen such that $A$ is invertible, the vector of polynomial coefficients $\mathbf{p}$ can be computed as $A^{-1} \tilde{\mathbf{p}}$.
	This polynomial can then be evaluated at $b=2^{n/k}$
 by taking the inner product of $\mathbf{p}$ and $\mathbf{e}_{2^{n/k}}$, yielding  the integer product $p$.
	Putting it all together, the integer value $p = xy$ can be computed as
	\begin{equation}
		xy = \mathbf{e}_{2^{n/k}}^{\intercal} A^{-1} (A\mathbf{x} \circ A\mathbf{y}).
		\label{eq:toom-cook-full}
	\end{equation}
	For clarity of intuition, we provide an explicit instantiation of this expression for $k=2$---corresponding to the Karatsuba algorithm---in Appendix~\ref{app:karatsuba-toom-cook}.
	
	\subsection{Quantum multiplication via the QFT}
	
	Many quantum circuits for multiplication mirror classical ones, with clever optimizations to account for unitarity~\cite{vedral_quantum_1996, zalka_fast_1998, cuccaro_new_2004, draper_logarithmic-depth_2006, haner_factoring_2017, gidney_factoring_2018, gidney_asymptotically_2019}.
	Here, we describe an early result by Draper, which stands out as an example of a quantum arithmetic circuit that has no classical analogue~\cite{draper_addition_2000}.
	Consider the quantum Fourier transform (QFT) of the output register in Eq.~\ref{eq:cq_unitary}, before and after the application of the classical-quantum multiply unitary $\mathcal{U}_{c \times q} (a)$,
	\begin{align*}
		\ket{x} \ket{w} &\xrightarrow{(\mathbb{I} \otimes \mathsf{QFT})} \sum_z \ket{x} \exp \left(2\pi i wz / 2^{2n} \right) \ket{z} \\
		 \ket{x} \ket{w + ax} &\xrightarrow{(\mathbb{I} \otimes \mathsf{QFT})} \sum_z \ket{x} \exp \left(2\pi i (w+ax) z / 2^{2n} \right) \ket{z}.
	\end{align*}
	The two only differ by a phase, $2\pi axz / 2^{2n}$, on each element of the superposition over $\ket{z}$.
	This suggests the following strategy for implementing $\mathcal{U}_{c \times q}(a)$: (i) apply a QFT to the output register, (ii) apply the diagonal unitary 
	\begin{equation}
		\tilde{\mathcal{U}}_{c \times q} (a) \ket{x} \ket{z} = \exp(2\pi i axz / 2^{2n}) \ket{x} \ket{z}
		\label{eq:cq_phase}
	\end{equation}
	and then (iii) apply an inverse QFT.
In effect, this decomposes the original unitary as: 
 $\mathcal{U}_{c \times q} (a) = (\mathbb{I} \otimes \mathsf{IQFT}) \tilde{\mathcal{U}}_{c \times q} (a) (\mathbb{I} \otimes \mathsf{QFT})$.
	For the quantum-quantum multiplication unitary $\mathcal{U}_{q \times q}$ [Eq.~\ref{eq:qq_unitary}], the strategy is identical, with the diagonal unitary being,
	\begin{equation}
		\tilde{\mathcal{U}}_{q \times q} \ket{x} \ket{y} \ket{z} = \exp(2\pi i xyz / 2^{2n}) \ket{x} \ket{y} \ket{z}
		\label{eq:qq_phase}
	\end{equation}

To date,  the phase rotation unitaries, $\tilde{\mathcal{U}}_{c \times q}(a)$ and $\tilde{\mathcal{U}}_{q \times q}$, have been constructed via the binary schoolbook decomposition of the products $xz$ and $xyz$:
	\begin{align}
		\exp(2\pi i axz / 2^{2n}) &= \prod_{ik} \exp(2\pi i a 2^{i+k} x_i z_k/ 2^{2n}) \label{eq:fourier-schoolbook} \\
		\exp(2\pi i xyz / 2^{2n}) &= \prod_{ijk} \exp(2\pi i 2^{i+j+k} x_i y_j z_k / 2^{2n}).\label{eq:fourier-qq-schoolbook}
	\end{align}
	Since the products $x_i z_k$ and $x_i y_j z_k$ are over binary values, they are 0 if any of the bits are zero (corresponding to no phase shift) and 1 if all the bits are 1.
	Thus,  the phases can be implemented as a series of either singly-controlled  or doubly-controlled  phase rotations, for $\tilde{\mathcal{U}}_{c \times q}(a)$ and  $\tilde{\mathcal{U}}_{q \times q}$, respectively.
For the classical-quantum case, the number of controlled phase rotations is $\mathcal{O}(n^2)$, while for the quantum-quantum case, one has to perform $\mathcal{O}(n^3)$ phase rotations; the difference comes from the fact that one can scale the entire phase rotation  by $a$ since it is a classical value.\footnote{In both cases, the quantum Fourier transforms (and their inverses) can be performed in $\mathcal{O}(n^2)$ time with no ancilla qubits. Thus, they do not affect the asymptotic complexity.}
In both cases, the number of ancilla qubits required is zero. 

	\section{Sub-quadratic quantum multiplication without ancillas}
	\label{sec:algorithm}	
	
	Our main result is that it is possible to implement the phases of Eqs.~\ref{eq:cq_phase} and~\ref{eq:qq_phase} via a construction based on the Toom-Cook algorithm, without ancilla qubits.
	As a corollary, we also find an algorithm for the exact quantum Fourier transform using a sub-quadratic number of gates and no ancillas.
	Taken together, these results yield ancilla-free algorithms for fast  quantum integer multiplication, whose explicit asymptotic scalings for are shown in Table~\ref{tab:scaling}.
	
	\begin{table}
		\begin{subtable}{0.5\textwidth}
			\centering
			\begin{tabular}{|c|c|}
				\hline 
				Algorithm & Gate count $\mathcal{O}(n^{\log_k (2k-1)})$ \\ 
				\hline 
				Eq.~\ref{eq:fourier-schoolbook} & $\mathcal{O}(n^2)$ \\ 
				\hline 
				$k=2$ & $\mathcal{O}(n^{1.58\cdots})$ \\ 
				\hline
				$k=5$ & $\mathcal{O}(n^{1.37\cdots})$ \\ 
				\hline
				$k=8$ & $\mathcal{O}(n^{1.30\cdots})$ \\ 
				\hline
			\end{tabular}
			\caption{Classical-quantum multiplication}
		\end{subtable}
		\begin{subtable}{0.5\textwidth}
			\centering
			\begin{tabular}{|c|c|}
				\hline 
				Algorithm & Gate count $\mathcal{O}(n^{\log_k (3k-1)})$ \\ 
				\hline 
				Eq.~\ref{eq:fourier-qq-schoolbook} & $\mathcal{O}(n^3)$ \\ 
				\hline 
				$k=2$ & $\mathcal{O}(n^2)$ \\ 
				\hline 
				$k=3$ & $\mathcal{O}(n^{1.77\cdots})$ \\ 
				\hline 
				$k=6$ & $\mathcal{O}(n^{1.55\cdots})$ \\ 
				\hline 
				$k=9$ & $\mathcal{O}(n^{1.46\cdots})$ \\ 
				\hline 
			\end{tabular} 
			\caption{Quantum-quantum multiplication}
		\end{subtable}
		\caption{
			\textbf{Asymptotic scaling of gate counts for various selected $k$.}
			Note that the constant factors (which are not visible in big-$\mathcal{O}$ notation) become worse as $k$ increases, resulting in a tradeoff.
			The values of $k$ shown here are practical for input integers up to a few thousand bits (see Sec.~\ref{sec:applications}).
		}
		\label{tab:scaling}
	\end{table}
	
	\textbf{Phase rotation algorithms---}Consider a generalization of the phase rotation of Eq.~\ref{eq:cq_phase}, which we denote the \textsf{PhaseProduct} and define as follows:
	\begin{equation}
		\mathsf{PhaseProduct}\left(\phi\right)   \ket{x} \ket{z} =  \exp(i \phi x z)  \ket{x} \ket{z}.
		\label{eq:phase-prod}
	\end{equation}
	Note that $\mathsf{PhaseProduct}(2\pi a/2^{2n}) = \tilde{\mathcal{U}}_{c \times q}(a)$.
	Our goal is to decompose the phase $\phi x z$ into a sum of many phases, that are each  easier to implement.
	At first glance, the Toom-Cook construction does not seem helpful:   it involves both recursion and a complicated linear algebra expression (Eq.~\ref{eq:toom-cook-full}) that does not obviously decompose into a single sum.
	Our key  observation is that $\mathbf{e}_{2^{n/k}}^\intercal$ and $A^{-1}$ are classically-known values, and thus a vector $\vec{\phi} = \phi \mathbf{e}_{2^{n/k}}^\intercal A^{-1}$ can be pre-computed.
	Using this in Eq.~\ref{eq:toom-cook-full} for the product $xz$ yields
	\begin{equation}
		\phi x z = \sum_{\ell=0}^{q-1}\phi_\ell (A \mathbf{x})_\ell (A \mathbf{z})_\ell ,
		\label{eq:phase-decomp}
	\end{equation}
	where $\phi_\ell$ are the elements of $\vec{\phi}$.
	To implement this decomposition, the only linear algebra that needs to be done by the quantum circuit is the computation of $(A \mathbf{x})_\ell$ and $(A \mathbf{z})_\ell$.
 As in the classical Toom-Cook algorithm, by  choosing  $w_\ell$ appropriately, these values can be computed using only a small number of additions.
	Furthermore, each term of the sum in Eq.~\ref{eq:phase-decomp} has the form of a \textsf{PhaseProduct} itself, but applied to values of size roughly $n/k$---thus, the decomposition can be applied recursively (Fig.~\ref{fig:circuit}).
	Doing so yields a decomposition of $\phi x z$ into $\mathcal{O}(n^{\log_k q})$ phase rotations that can each be implemented in a constant number of gates, plus an asymptotically negligible number of gates to compute (and uncompute) each $(A \mathbf{x})_\ell$ and $(A \mathbf{z})_\ell$.
	Thus, the total number of gates is $\mathcal{O}(n^{\log_k q})$ matching the asymptotic complexity of the classical Toom-Cook algorithm.

	Naively, ancilla registers would be needed to store the values $(A \mathbf{x})_\ell$ and $(A \mathbf{z})_\ell$.
	We use the fact that addition is reversible to instead overwrite parts of the input registers $\ket{x}$ and $\ket{z}$ with those values---for example, to compute $\ket{x_0 + x_1}$ we may simply sum $x_1$ into the part of the $\ket{x}$ register already containing $x_0$.\footnote{The idea of using the reversibility of addition to avoid allocating extra ancilla registers was proposed in~\cite{gidney_asymptotically_2019} in the context of the Karatsuba algorithm, however the overflow bits were not considered in that work because $\mathcal{O}(n)$ ancilla qubits were already otherwise being used.\label{foot:overwrite}}
	Because $(A \mathbf{x})_\ell$ and $(A \mathbf{z})_\ell$ may take on slightly larger values than the registers they overwrite, there will be some ``overflow" carry bits.
    We avoid having to store those in ancillas by devising a way to directly implement the part of the \textsf{PhaseProduct} that involves the overflow bits, without ever explicitly storing their value in ancillas (see Appendix~\ref{app:overflow-phase}).
    Thus, no ancilla qubits are needed for any part of the algorithm.
	
	Next, we discuss quantum-quantum multiplication, which requires only a small modification.
	We define an operation analogous to  \textsf{PhaseProduct}, but which operates on three registers instead of two:
	\begin{align}		\mathsf{PhaseTripleProduct}\left(\phi\right)  &\ket{x} \ket{y} \ket{z} \nonumber \\
		=  \exp(i \phi x y z)  &\ket{x} \ket{y} \ket{z}.
		\label{eq:triple-phase-prod}
	\end{align}
	In a less constrained setting, no special algorithm is required for the product of three integers---one would simply take the product of two and then multiply the third  by the result.
	However, we cannot make use of intermediate values like this in the phase;  instead, the entire product must be decomposed into a single sum.
	To this end, we introduce a modified version of the Toom-Cook construction: we define three polynomials $x(b)$, $y(b)$, $z(b)$, and compute their product $p(b) = x(b) y(b) z(b)$  by evaluating the polynomials at a set of points $\{w_\ell\}$, finding their pointwise product, and then interpolating. 
	The only difference from the conventional Toom-Cook construction  
is that $p(b)$ will now have degree $3(k-1)$ and is uniquely determined by $q=3k-2$ points.
	Aside from the slightly larger value of $q$, the structure of the quantum algorithm is the  same.
	We may decompose $\phi xyz$ as
	\begin{equation}
		\phi x y z = \sum_{\ell=0}^{q-1} \phi_\ell (A \mathbf{x})_\ell (A \mathbf{y})_\ell (A \mathbf{z})_\ell
	\end{equation}
	and recursively apply $\mathsf{PhaseTripleProduct}$ to each of the terms on the right hand side.
	The larger value of $q$ leads to  a slightly worse asymptotic scaling, but still outperforms $\mathcal{O}(n^2)$ for $k\geq 3$ [Table~\ref{tab:scaling}].
	Again, overwriting parts of the input register to compute the linear combinations, and performing the phase rotations corresponding to the overflow bits directly, our algorithm requires no ancilla qubits.

	Before moving on, we note two generalizations to $\mathcal{U}_{c\times q} (a)$ and $\mathcal{U}_{q\times q}$.
	Firstly, we observe that a classical-quantum-quantum product $\mathcal{U}_{c \times q\times q} (a)$ may be implemented via our quantum-quantum multiplication algorithm with no extra cost, by simply scaling $\phi$ by $a$.
	Secondly, we note that the classical value $a$ does not need to be an integer---it can be an arbitrary real number.
	Of course, for an exact answer, the product must be representable by the output qubits, but even if it is not, the final state will be a superposition of values near the correct value (see Section~5.2.1 of~\cite{nielsen_quantum_2011}).
	
	\textbf{Fast exact quantum Fourier transform---}Our algorithm for \textsf{PhaseProduct} can be directly used to implement an exact quantum Fourier transform with the same sub-quadratic gate count as classical-quantum multiplication, and again no ancilla qubits.
	It is known that the exact quantum Fourier transform with modulus $2^n$,  $\mathsf{QFT}_{2^n}$, can be implemented via the following three steps for any positive integer $m < n$~\cite{cleve_fast_2000}:
	\begin{enumerate}
		\item Apply $\mathsf{QFT}_{2^{m}}$ to the first $m$ qubits
		\item Apply $\mathsf{PhaseProduct}(2 \pi/2^{n})$ to $\ket{x} \ket{y}$, where $x$ is the value of the first $m$ qubits and $y$ is the value of the remaining $n-m$ qubits
		\item Apply $\mathsf{QFT}_{2^{n-m}}$ to the final $n-m$ qubits.
	\end{enumerate}
	Setting $m=1$ corresponds to the standard construction for the QFT, using $\mathcal{O}(n^2)$ gates.
	By setting $m=n/2$ and using recursion to perform steps 1 and 3, the asymptotic runtime can be improved.
	Prior to this work, the proposed implementation for step 2 was to directly compute the product $xy$ into an ancilla register, via a classical fast multiplication algorithm that has been compiled into a quantum circuit (at the cost of many ancillas).
 	Single-qubit phase rotations are then applied to the bits of the product, after which it must be uncomputed. 
	Instead, using our algorithm for \textsf{PhaseProduct} in step 2 yields an algorithm for the exact QFT whose asymptotic runtime matches that of our implementation of \textsf{PhaseProduct} without any ancillas.
	
	\section{Space-time trade-offs}
	\label{sec:spacetime}

In the previous sections, we have shown that sub-quadratic quantum multiplication can be performed with zero ancillas. 
However, we find that the use of just a few ancillas is helpful for improving the practical performance of our algorithm.  
As an example, explicitly storing the overflow bits eliminates the extra work required to implement that portion of the phase rotation (see Appendix~\ref{app:overflow-phase}), and requires at most a few dozen extra qubits for inputs of even several thousand bits.
In what follows, we discuss three other space-time trade-offs: (i) improving the base case of the recursion, (ii) implementing arbitrary phase rotations with low overhead, and (iii) achieving sub-linear circuit depth. 
	
	\textbf{Base case optimization---}One instance where a few extra ancillas can improve performance considerably is in the base case of \textsf{PhaseTripleProduct}.
 In particular, when the recursion has reached a sufficiently small size, denoted $n_\mathrm{base}$,  
the phase rotation is performed directly.
	For \textsf{PhaseTripleProduct},  implementing the base case via Eq.~\ref{eq:fourier-qq-schoolbook} requires $n_\mathrm{base}^3$ doubly-controlled phase gates.
	Instead, we propose a ``semi-digital'' implementation which requires only  $\sim n_\mathrm{base}^2$ gates and is structured as follows.
	To implement the base-case phase rotation of $\phi' x' y' z'$, start by explicitly computing $\ket{x'y'}$ into an ancilla register using a standard multiplication circuit. 
 Next, use $2n_\mathrm{base}^2$ controlled phase rotations (off of the ancilla register) to apply a phase of $\phi (x'y') z'$.
Finally, uncompute the ancilla register.
This procedure requires at most $2n_\mathrm{base}$ ancilla qubits.\footnote{We may even reduce the qubit count further by reusing a smaller number of qubits to compute and uncompute small chunks of $x'y'$ in sequence.}
 We note that for \textsf{PhaseProduct}, the above optimization is less important, since its base case can be implemented directly with Eq.~\ref{eq:fourier-schoolbook} using only $n_\mathrm{base}^2$ singly controlled phase rotations.
	
	\textbf{Implementing $CR_\phi$ gates with low overhead via a phase gradient---}In settings where only a discrete gate set is native (such as in the logical qubits of quantum error correcting codes), arbitrary phase rotations cannot be directly implemented~\cite{heeres_implementing_2017, acharya_suppressing_2023, bluvstein_quantum_2022}.
	However, the overall cost of synthesizing them can be  dramatically reduced by using a few ancilla qubits to store a so-called \textit{phase gradient} state: $\ket{\Phi} = \sum_{\omega = 0}^{2^m - 1} e^{-2 \pi i \omega / 2^m} \ket{\omega}$~\cite{kitaev_classical_2002, gidney_turning_2016, gidney_halving_2018, nam_approximate_2020}.\footnote{For a pedagogical exposition of the use of phase gradients to implement the quantum Fourier transform, see~\cite{gidney_turning_2016}.}
	Here, $m = \lceil \log_2 1/\eta \rceil$  where $\eta$ is the desired precision  and $\omega$ are integers.
	To apply an arbitrary phase $\phi$, one can simply use a quantum addition circuit to increment the $\omega$ register by $a = \lceil 2^m \phi/2\pi \rfloor$,  where $\lceil \cdot \rfloor$ denotes rounding to the nearest integer. This yields a phase shift of $\phi$, up to precision $\eta = 2^{-m}$:
	\begin{equation}
		\sum_{\omega = 0}^{2^m - 1} e^{-2 \pi i \omega / 2^m} \ket{\omega + a} = e^{2 \pi i a / 2^m} \ket{\Phi} \approx e^{i\phi} \ket{\Phi}.
	\end{equation}
	Thus, the $CR_\phi$ gates of our algorithm can be implemented via doubly-controlled addition circuits. 
	Because $\ket{\Phi}$ is an eigenstate of the addition circuit, it is not destroyed by this process and can be reused an arbitrary number of times.

	As discussed earlier, the base case of both \textsf{PhaseProduct} and \textsf{PhaseTripleProduct} can be implemented via $\sim n_\mathrm{base}^2$ $CR_\phi$ gates.
A naive use of the phase gradient state to implement these $\sim n_\mathrm{base}^2$ $CR_\phi$ gates yields a total cost of $\sim n_\mathrm{base}^2 \log( 1/\eta)$ Toffoli gates.
Instead, by again applying the semi-digital optimization introduced above, one can convert the $\sim n_\mathrm{base}^2$ $CR_\phi$ gates to $\sim n_\mathrm{base}^2$ Toffoli gates and $\sim n_\mathrm{base}$ $R_\phi$ gates.
Putting everything together, this technique implements the required rotations in \textsf{PhaseProduct} and \textsf{PhaseTripleProduct} using a total of only $ \sim (n_\mathrm{base}^2 + n_\mathrm{base} \log(1/\eta))$ Toffoli gates (plus some Clifford gates).
Since one expects $\log_2 (1/\eta) < n_\mathrm{base}$, this represents a considerable improvement in the overall cost. 
	
	\textbf{Circuit depth---}The recursive tree structure of our algorithms lends itself well to parallelization, as multiple branches of the tree can be performed simultaneously (Fig.~\ref{fig:circuit}).
	Since each input register is divided into $k$ parts, it is possible to perform $k$ recursive calls in parallel by storing a different linear combination in each part.
	By parallelizing in this way, \textsf{PhaseProduct} can be implemented in depth $\mathcal{O}(n^{\log_k 2})$ and \textsf{PhaseTripleProduct} in $\mathcal{O}(n^{\log_k 3})$.
	These scalings are sublinear for $k>2$ and $k>3$ respectively, and asymptotically the exponents can be made arbitrarily close to zero by increasing $k$.
	
	Achieving sublinear depth for \textsf{PhaseProduct} and \textsf{PhaseTripleProduct} does seem to require a small (sublinear) number of ancilla qubits, as sublinear depth addition (which is needed to compute the linear combinations $(A\bm{x})_\ell$) seems to require ancillas~\cite{takahashi_fast_2008}.\footnote{It is also not immediately clear how to parallelize the dirty qubit techniques of Appendix~\ref{app:overflow-phase} into sublinear depth.}
	However, surprisingly, the space bottleneck for sublinear-depth multiplication seems not to to be the phase rotation algorithms, but rather the quantum Fourier transforms on the output register.
	Indeed, to our knowledge the most space-efficient sublinear-depth (approximate) quantum Fourier transform requires $\mathcal{O}(n \log n)$ ancillas~\cite{cleve_fast_2000}.
	With this in mind, we formalize our result regarding the depth of quantum multiplication in the following claim, leaving the cost of the quantum Fourier transform abstract to allow for a choice of implementation:
	
	\begin{restatable}{claim}{depthclaim}
		\label{claim:depth}
		Suppose there exists a circuit for the quantum Fourier transform on $n$ qubits having depth $D_\mathsf{QFT}(n)$, ancilla count $A_\mathsf{QFT}(n)$ and gate count $G_\mathsf{QFT}(n)$.
		Then, there exist circuits for quantum integer multiplication of $n$-qubit values (Eqs.~\ref{eq:cq_unitary} and \ref{eq:qq_unitary}) with depth $\mathcal{O}(n^\epsilon) + 2 D_\mathsf{QFT}(n)$ and ancilla count $\max(\mathcal{O}(n/\log n), A_\mathsf{QFT}(n))$ using $\mathcal{O}(n^{1+\epsilon}) + 2G_\mathsf{QFT}(n)$ gates, for any $\epsilon > 0$.
	\end{restatable}
	
	\begin{proof}
		See Appendix~\ref{app:depth-proof}.
	\end{proof}
	
	A challenge in achieving the low depth and low qubit count of Claim~\ref{claim:depth} is that $k$ linear combinations $(A \mathbf{x})_\ell$ must be computed in-place, each overwriting part of the $\ket{x}$ register---but in general each $(A \mathbf{x})_\ell$ depends on \textit{all} the bits of $\ket{x}$, including those being overwritten by other values!
	While we show in the proof of Claim~\ref{claim:depth} that it is indeed possible to perform the desired computation in-place in polylogarithmic depth, in practice this step could lead to large constant factors if done without care.
	Thus, in Appendix~\ref{app:parallel-sequences}, we propose explicit sequences of parallel operations designed by hand for a few practically-relevant values of $k$, with the $w_\ell$ chosen carefully so that the linear combinations can be computed in-place via a small number of additions.
	Finding such efficient sequences for more values of $k$ may be important if the low-depth constructions are to be used in practice.
	
	Finally, we note that with the use of a superlinear number of qubits, we may parallelize all $q$ recursive calls into a single layer.
	Doing so yields an asymptotic depth of $\mathcal{O}(\log^2 n)$, and a total qubit count of $\mathcal{O}(n^{1+\epsilon})$ for any $\epsilon > 0$.
	It is already known that multiplication can be performed in depth $\mathcal{O}(\log^2 n)$ using $\mathcal{O}(n \log n \log \log n)$ qubits via a parallel quantum version of the Sch\"onhage-Strassen algorithm~\cite{nie_quantum_2023}; however, that algorithm has large constant factors, and it may be the case that for certain values of $n$ and $\epsilon$ our circuit uses fewer qubits and/or less depth in practice.
	The exploration of the practical costs of this extremely low-depth construction may be an interesting direction for future work.

 	\section{Applications}
	\label{sec:applications}
We now explore the use of our fast quantum multiplication algorithm  as a subroutine in two specific applications: Shor's algorithm for factoring, and a cryptographic proof of quantum computational advantage.
 	While we leave the explicit construction and optimization of quantum circuits to future work, here, we perform careful estimates of the resources required and find promising gate counts along with a dramatic reduction in the number of ancilla qubits (see Tables~\ref{tab:cq-2048-costs} and~\ref{tab:x2modN-costs}).\footnote{The code used to estimate these gates counts is available online: \url{https://zenodo.org/doi/10.5281/zenodo.10871109}}
The constructions we have presented so far implement multiplication over all integers.
	Many applications, including both that we discuss in this section, instead require multiplication over the integers modulo some $n$-bit integer $N$.

	\textbf{Modular multiplication---}It is straightforward to use generic multiplication to implement modular multiplication~\cite{montgomery_modular_1985}, but doing so increases the size of the circuit considerably and also requires the allocation of a $2n$-bit register to store the full product.\footnote{In some cases, Zalka's coset representation allows approximate quantum modular multiplication to be performed with essentially no overhead compared to standard multiplication~\cite{zalka_shors_2006}; fast multipliers like the one presented in this paper and the Karatsuba circuits of~\cite{gidney_windowed_2019, gidney_asymptotically_2019} seem to have a structure that is fundamentally incompatible with Zalka's trick.}
	Here, we propose to utilize an alternate strategy that leverages the cyclic nature of the quantum phase to perform the modulo operation automatically~\cite{kahanamoku-meyer_classically_2022}.
	That is, we replace Eqs.~\ref{eq:cq_phase} and~\ref{eq:qq_phase} with
	\begin{align}
		\tilde{\mathcal{U}}_{c \times q}' (a) \ket{x} \ket{z} &= \exp(2\pi i axz / N) \ket{x} \ket{z} \\
		\tilde{\mathcal{U}}_{q \times q}' \ket{x} \ket{y} \ket{z} &= \exp(2\pi i xyz / N) \ket{x} \ket{y} \ket{z}
	\end{align}
	such that multiples of $N$ in the product become multiples of $2\pi$ in the phase.
	The only challenge with this strategy is that performing (inverse) quantum Fourier transforms modulo arbitrary $N$ is too expensive.
	Instead, we propose to still perform the QFT modulo $2^n$, which will yield a state heavily weighted on the binary fraction that most closely approximates the rational value $ax/N$ (classical-quantum multiplication) or $xy/N$ (quantum-quantum multiplication); this binary fraction then uniquely identifies the integer value.
  From quantum phase estimation, it is known that an output register of $n+\mathcal{O}(\log(1/\eta))$ qubits is sufficient to ensure that the final state is within $\eta$ of $ax/N$ or $xy/N$, respectively~\cite{nielsen_quantum_2011}.
	
	\subsection{Classical-quantum multiplication: Shor's algorithm}

	Perhaps the most obvious application of our result is to Shor's algorithm for integer factorization, which can be implemented via $\mathcal{O}(n)$ controlled, in-place classical-quantum modular multiplications.
 In particular, it is possible to implement these multiplications to precision $\eta$ using a total of $2n + \mathcal{O}(\log(1/\eta))$ qubits (see Appendix~\ref{app:shor} for details).
 Thus, asymptotically, our technique yields a circuit for factoring  with $\mathcal{O}(n^{2+\epsilon})$ gates (for arbitrarily small $\epsilon$) and only $2n + \mathcal{O}(\log n)$ qubits.
 To our knowledge, this is by far the best qubit count of any factoring circuit with a sub-cubic number of gates.
 
 \begin{table*}
 	\begin{center}
 		\begin{tabular}{|c|c|c|c|c|c|}
 			\hline
 			\multirow{2}*{Algorithm} & Asymptotic & \multicolumn{3}{c|}{Gate count (millions)} & \multirow{2}*{Ancillas} \\ \cline{3-5} 
 			& scaling & Toffoli & $CR_\phi$ & $H$,$X$,CNOT & \\
 			\hline
 			\textbf{This work (standard QFT)} & $\bm{\mathcal{O}(n^{1.29})}^*$ & \textbf{0.6} & \textbf{0.3} & \textbf{1.9} & \textbf{79} \\
 			\textbf{This work (phase gradient QFT)} & $\bm{\mathcal{O}(n^{1.29})}^*$ & \textbf{0.9} & \textbf{0.1} & \textbf{3.2} & \textbf{80} \\
 			Karatsuba~\cite{gidney_windowed_2019} & $\mathcal{O}(n^{1.58})$ & 5.6 & --- & 34 & 12730 \\
 			Windowed~\cite{gidney_windowed_2019} & $\mathcal{O}(n^2 / \log^2 n)$ & 1.8 & --- & 2.5 & 4106 \\
 			Schoolbook~\cite{gidney_windowed_2019} & $\mathcal{O}(n^2)$ & 6.4 & --- & 38 & $1^{**}$ \\
 			\hline
 		\end{tabular}
 	\end{center}
 	
 	\caption{
 		\textbf{Circuit size estimates for one classical-quantum multiplication of 2048 bit numbers.}
 		Results from this work in bold, previous works in non-bold.
 		All estimates are in the ``abstract circuit model'' (no error correction or routing costs included).
 		Costs of previous works were computed using the Q\# code from~\cite{gidney_windowed_2019}.
 		Note that that code implements a regular multiplication rather than a multiplication $\bmod N$; using Zalka's coset representation of integers~\cite{zalka_shors_2006} the cost of modular multiplication is expected to be roughly the same as standard multiplication (except in the case of Karatsuba, for which it does not seem possible to apply Zalka's optimization and thus modular multiplication may be considerably more expensive).
 		Results for ``this work'' are reported for modular multiplication, and include the quantum Fourier transforms before and after the \textsf{PhaseProduct}, performed to a precision of $\eta = 10^{-12}$ per qubit.
 		Note also that the ``phase gradient QFT'' estimates here do not perform the arbitrary phase rotations in the base case of the \textsf{PhaseProduct} via phase gradient state, only the QFTs.
 		$^{*}$We vary the parameter $k$ throughout the recursion; this is the asymptotic scaling for $k=9$ which is the value at the top level of recursion.
 		$^{**}$An ancilla count of 2048 was reported in~\cite{gidney_windowed_2019} and by the associated code; via qubit reuse it should be possible to reduce this to 1.
 		We did not explore how qubit reuse could be applied to the other two constructions from~\cite{gidney_windowed_2019}.
 	}
 	\label{tab:cq-2048-costs}
 \end{table*}
	
	In order to get a sense of  circuit sizes in practice, we perform a detailed analysis of the gate and qubit counts for the out-of-place multiplication of 2048-bit inputs (which is  used to construct the in-place multiplications described above).
Specifically, we build the recursive tree (Fig.~\ref{fig:circuit}) and carefully tabulate the number and type of gates in each layer.
Via a careful choice of the points $w_\ell$, we are able to devise efficient addition sequences that minimize the overhead of forming linear combinations such as $(A\mathbf{x})_\ell$.
Moreover, at each layer, we  explicitly optimize over $k$, the number of pieces into which we divide the inputs (see e.g.~Sec.~\ref{subsec:toom-cook} and Table~\ref{tab:scaling}).
Interestingly, we find that it is sometimes optimal to use different values of $k$ even at the same level of recursion (see Appendix~\ref{app:estimation_details}).

Our results are depicted in
Table~\ref{tab:cq-2048-costs}.
For context, we compare to recent work which implements quantum circuits for classical-quantum multiplication, including via the sub-quadratic Karatsuba algorithm~\cite{gidney_windowed_2019}.\footnote{Note that the estimates from~\cite{gidney_windowed_2019} are for non-modular arithmetic.}
As discussed in the previous section, we have used a small number of extra ancillas to reduce the constant factors on the gate count.  
Even with this optimization, the ancilla counts are orders of magnitude less than comparable implementations!
	Moreover, the gate counts are very promising, although impossible to compare directly to alternate strategies without compiling to a common native gate set.
	We note that the majority of the $CR_\phi$ gates occur not in the \textsf{PhaseProduct} but in the QFTs performed on either side of it.
	As shown in Table~\ref{tab:cq-2048-costs}, implementing these QFTs using a phase gradient state   reduces the total number of $CR_\phi$ gates considerably.
	
	\subsection{$x^2 \bmod N$: cryptographic proofs of quantum computational advantage}
	
	Recently, much excitement has centered on the experimental realization of quantum computational advantage---specifically, the performance of random sampling tasks that are infeasible for even the world's fastest classical supercomputers~\cite{arute_quantum_2019, zhong_quantum_2020, wu_strong_2021, zhu_quantum_2022, morvan_phase_2023}.
	So far, such experiments have had the subtle characteristic that \textit{checking} the results classically is as hard as, if not harder than, generating the solution.
	However, recent theoretical work has explored the use of cryptography to design ``test of quantum computational advantage'' protocols which are classically efficiently verifiable, yet remain hard to spoof~\cite{brakerski_cryptographic_2021, brakerski_simpler_2020, kahanamoku-meyer_classically_2022}.
	When instantiated with quantum-secure cryptographic assumptions, this class of protocols is also potentially useful for other cryptographic tasks such as certifiable random number generation, remote state preparation, and classical delegation of quantum computations to untrusted devices~\cite{brakerski_cryptographic_2021, gheorghiu_computationally-secure_2019, mahadev_classical_2018, brakerski_simple_2023, natarajan_bounding_2023}.
	Because much of the underlying cryptography involves arithmetic, our algorithm for multiplication is directly relevant to the protocols' implementation.
	
	Here, we focus specifically on a recent proposal for a test of quantum computational advantage in which the quantum computer must evaluate the function $f(x) = x^2 \bmod N$ on a superposition of inputs $x$~\cite{kahanamoku-meyer_classically_2022}.
	This function, and the protocol it is used in, permit a number of helpful optimizations.
	For our multiplication algorithm, computing a square instead of a general quantum-quantum multiplication is less costly because the linear combinations like $(A\mathbf{x})_\ell$ need only be computed on two registers ($x$ and $z$) instead of three.
	The protocol also has a built-in measurement based uncomputation scheme, which allows ancilla qubits to be discarded by measuring them in the Hadamard basis and recording the measurement results, avoiding the need for explicit uncomputation of garbage bits~\cite{kahanamoku-meyer_classically_2022}.
	
	\begin{table*}
		\begin{center}
			\begin{tabular}{|c|c|c|c|c|c|c|}
				\hline
				\multirow{2}*{Algorithm} & Asymptotic & \multicolumn{4}{c|}{Gate count (millions)} & Total \\ \cline{3-6}
				& scaling & Toffoli & $CR_\phi$ & Measmt. & $H$,$X$,CNOT & qubits \\
				\hline
				\textbf{This work, ``fast''} & $\bm{\mathcal{O}(n^{1.49})}^*$ & \textbf{0.6} & \textbf{0.7} & \textbf{0.4} & \textbf{0.8} & \textbf{2937} \\
				\textbf{This work, ``balanced''} & $\bm{\mathcal{O}(n^{1.49})}^*$ & \textbf{0.7} & \textbf{0.7} & \textbf{0.4} & \textbf{1.1} & \textbf{2140} \\
				\textbf{This work, ``narrow''} & $\bm{\mathcal{O}(n^{1.55})}^*$ & \textbf{1.9} & \textbf{1.4} & \textbf{1.0} & \textbf{2.7} & \textbf{1583} \\ \hline
				Prev. Fourier 1~\cite{kahanamoku-meyer_classically_2022} & $\mathcal{O}(n^3)$ & --- & \textcolor{purple}{$539^{**}$} & --- & --- & 1025 \\
				Prev. Fourier 2~\cite{kahanamoku-meyer_classically_2022} & $\mathcal{O}(n^2 \log n)$ & --- & \textcolor{purple}{35} & --- & --- & 2062 \\ \hline
				``Digital'' Karatsuba~\cite{kahanamoku-meyer_classically_2022} & $\mathcal{O}(n^{1.58})$& 1.6 & --- & 1.1 & 1.6 & \textcolor{purple}{6801} \\
				``Digital'' Schoolbook~\cite{kahanamoku-meyer_classically_2022} & $\mathcal{O}(n^2)$ & 3.5 & --- & 2.2 & 2.9 & \textcolor{purple}{4097} \\
				\hline
			\end{tabular}
		\end{center}
		
		\caption{
			\textbf{Circuit size estimates for $x^2 \bmod N$ in the context of the proof of quantumness protocol, for 1024-bit $N$.}
			Results from this work in bold, previous results in non-bold.
			All estimates are in the ``abstract circuit model'' (no error correction or routing costs included).
			$^{*}$We vary the parameter $k$ throughout the recursion; these are the asymptotic scalings for the values of $k$ used at the top level of the recursions ($k=6$ for ``narrow'' and $k=8$ for ``balanced'' and ``fast'').
			$^{**}$These gates are doubly, rather than singly, controlled $R_\phi$ gates.
		}
		\label{tab:x2modN-costs}
	\end{table*}

 We perform analogous estimates
of the gate and qubit
counts for computing $f(x)$ with 1024-bit $x$ and $N$.\footnote{This is a problem size we expect to be infeasible for modern supercomputers: the protocol's hardness is based on the hardness of factoring $N$, and the largest publicly-known factorization of an integer without a special form is of length 829 bits~\cite{boudot_state_2022}.}
We explore three constructions (``fast'', ``balanced'', and ``narrow''), with varying trade-offs between qubit and gate count, and our results are shown in Table~\ref{tab:x2modN-costs}.
The ``fast'' version of our circuits aggressively uses the measurement-based uncomputation scheme to reduce gate counts at the expense of a moderate number of ancillas.
    The ``balanced'' version computes and uncomputes the sums in-place, requiring a few more gates but many fewer ancillas.
    Finally, the ``narrow'' version never stores the entire output and instead reuses a register of $n/2$ qubits twice, further reducing qubit counts at the expense of a few more gates.
 
	We observe that all three constructions reduce the total qubit counts substantially when compared to previous ``digital'' implementations of the Karatsuba and Schoolbook algorithms.
	When compared with two other algorithms optimized for qubit count (at the expense of gate count), our algorithms achieve similarly low qubit counts while reducing gate counts by orders of magnitude.
	
	\section{Outlook}
	
	The optimization of quantum circuits for performing arithmetic on superpositions of inputs has been the subject of study for decades.
	In this work, we have introduced a technique for performing quantum multiplication in a sub-quadratic number of gates  with zero ancillas.
	
	Our results open the door to a number of intriguing open questions. 
	First, can factoring be performed in under $\mathcal{O}(n^3)$ gates  with even fewer than the $2n + \mathcal{O}(\log n)$ qubits reported here?
	One obvious path towards achieving this would be to devise a new way to implement modular arithmetic in our algorithm without requiring $\mathcal{O}(\log n)$ extra qubits~\cite{zalka_shors_2006}.
 Another path is to apply our multiplication algorithm to a space-efficient version of Regev's recent fast factoring algorithm~\cite{regev_efficient_2023, ragavan_space-efficient_2024}; however, it seems that further space optimization would be required.
 Second, we note that there are likely more clever constructions for implementing the controlled arbitrary phase rotations of the base case, which avoid the need to compute and uncompute an intermediate ancilla register.
	Finally, it is worth exploring whether our strategy is compatible with even faster classical multiplication algorithms, such as the $\mathcal{O}(n \log n \log \log n)$ Schonhage-Strassen algorithm~\cite{schonhage_schnelle_1971, knuth_art_1998}.

	The authors would like acknowledge the insights of and discussions with Tanuj Khattar, Craig Gidney, Isaac Chuang, Seyoon Ragavan, Katherine van Kirk, and John Blue.
 This work was supported by the NSF QLCI Award OMA-2016245 and the NSF STAQ II program. 
 N.Y.Y. acknowledges support from a Simons Investigator Award.

	\bibliography{references}
	\bibliographystyle{gregbib}
	
	\clearpage
	\appendix
	
	\section{The Karatsuba algorithm in the notation of Section~\ref{subsec:toom-cook}}
	\label{app:karatsuba-toom-cook}
	
	Here for pedagogical purposes we briefly describe the Karatsuba algorithm, of which Toom-Cook is a generalization, and show how it can be expressed in the linear algebra notation introduced in Section~\ref{subsec:toom-cook}.
	
	We begin by introducing the Karatsuba algorithm.
	It is an instance of Toom-Cook with $k=2$, meaning that the input values are split into two pieces, corresponding to the low and high halves of the bits: $x = 2^{n/2} x_1 + x_0$.
	A product $xy$ of values expressed in this way can be expanded to 
	\begin{equation}
		xy = 2^n x_1 y_1 + 2^{n/2} (x_1 y_0 + x_0 y_1) + x_0 y_0
		\label{eq:karatsuba-orig}
	\end{equation}
	which replaces the product of $n$-bit values with four products of $n/2$-bit values.
	The key to the Karatsuba algorithm is to observe that the quantity in parentheses can be written as
	\begin{equation}
		x_1 y_0 + x_0 y_1 = (x_0 + x_1)(y_0 + y_1) - x_1 y_1 - x_0 y_0
	\end{equation}
	and since the second two products there can be reused from the other terms of Eq.~\ref{eq:karatsuba-orig}, the entire product can be computed using only three products of size $n/2$ instead of four (plus a few additions).
	By recursively applying this technique, one may compute the product $xy$ in only $\mathcal{O}(n^{\log_2 3}) = \mathcal{O}(n^{1.58\cdots})$ operations.
	
	In the notation of Section~\ref{subsec:toom-cook}, Karatsuba corresponds to setting $k=2$ and using the evaluation points $w_\ell \in \{0, 1, \infty\}$.
	For the sake of gaining intuition, it may be helpful for the reader to work through Eq.~\ref{eq:toom-cook-full} using these values---the result should be the Karatsuba decomposition of the product $xy$:
	\begin{equation}
		2^n x_1 y_1 + 2^{n/2} \left[(x_0 + x_1) (y_0 + y_1) - x_1 y_1 - x_0 y_0 \right] + x_0 y_0
	\end{equation}
	Note that arriving at the correct expression is easiest if one considers the point $\infty$ as the unit fraction $1/0$; see 
	Appendix~\ref{app:fractional_w}.
	
	\section{Use of fractional $w_\ell$}
	\label{app:fractional_w}
	
	In both classical Toom-Cook and in our algorithms, it is common to take some of the evaluation points $w_\ell$ to be unit fractions $1/c$ for some $c$ (frequently one such point is $\infty$, which we may consider as a unit fraction with $c=0$).
	In this section we describe how careful maneuvering allows us to maintain a Vandermode matrix $A$ of integer values, even when using such unit fractional $w_\ell$.
	
	Consider a length-$q$ vector $\mathbf{e}_{1/c}$ (which is row $\ell$ of the matrix $A$ for some $\ell$):
	\begin{equation}
		\mathbf{e}_{1/c} = (1/c^{q-1}, 1/c^{q-2}, \cdots, 1/c^2, 1/c, 1)
	\end{equation}
	Then the $\ell^{\mathrm{th}}$ element of $\tilde{\mathbf{x}} = A\mathbf{x}$ is $\mathbf{e}_{1/c} \cdot \mathbf{x}$, where the length-$k$ vector $\mathbf{x}$ has been padded to length $q$ by inserting zero elements on the left.
	
	The Toom-Cook algorithm is built on the fact that the $\ell^\mathrm{th}$ element of the product vector $\tilde{\mathbf{p}}$ is the pointwise product of the $\ell^\mathrm{th}$ elements of $\tilde{\mathbf{x}}$ and $\tilde{\mathbf{y}}$, that is,
	\begin{equation}
		\mathbf{e}_{1/c} \cdot \mathbf{p} = (\mathbf{e}_{1/c} \cdot \mathbf{x})(\mathbf{e}_{1/c} \cdot \mathbf{y})
	\end{equation}
	Now, the key is to scale both sides of the expression by $c^{2k-2}$ and then distribute that constant across the $\mathbf{e}$ vectors:
	\begin{equation}
		c^{2k-2} \mathbf{e}_{1/c} \cdot \mathbf{p} = (c^{k-1} \mathbf{e}_{1/c} \cdot \mathbf{x})(c^{k-1} \mathbf{e}_{1/c} \cdot \mathbf{y})
	\end{equation}
	Now we define $\mathbf{e}''_{1/c} = c^{2k-2} \mathbf{e}_{1/c}$ and (noting that $q=2k-1$) see that is has the form
	\begin{equation}
		\mathbf{e}''_{1/c} = (1, c, c^2, \cdots, c^{q-2}, c^{q-1})
	\end{equation}
	all of which are integers.
	Similarly (dropping the elements that are multiplied by padded zeros of $\mathbf{x}$) we have
	\begin{equation}
		\mathbf{e}'_{1/c} = (1, c, c^2, \cdots, c^{k-2}, c^{k-1})
	\end{equation}
	
	In summary, we may perform Toom-Cook using slightly adjusted matrices $A'$ and $A''$ which have only integer elements, as follows.
	Let the first $m$ evaluation points $w_\ell$ be integers, and the remaining $q-m$ points be unit fractions $1/c_\ell$.
	Then we construct the $q \times k$ matrix
	\begin{equation}
		A' = \begin{pmatrix}
			w_0^{k-1} & w_0^{k-2} & \cdots & w_0 & 1 \\
			\vdots & \vdots & \vdots & \vdots & \vdots \\
			w_{m-1}^{k-1} & w_{m-1}^{k-2} & \cdots & w_{m-1} & 1 \\
			1 & c_m & \cdots & c_m^{k-2} & c_m^{k-1} \\
			\vdots & \vdots & \vdots & \vdots & \vdots \\
			1 & c_{q-1} & \cdots & c_{q-1}^{k-2} & c_{q-1}^{k-1} \\
		\end{pmatrix}
	\end{equation}
	and the $q \times q$ matrix
	\begin{equation}
		A'' = \begin{pmatrix}
			w_0^{q-1} & w_0^{q-2} & \cdots & w_0 & 1 \\
			\vdots & \vdots & \vdots & \vdots & \vdots \\
			w_{m-1}^{q-1} & w_{m-1}^{q-2} & \cdots & w_{m-1} & 1 \\
			1 & c_m & \cdots & c_m^{q-2} & c_m^{q-1} \\
			\vdots & \vdots & \vdots & \vdots & \vdots \\
			1 & c_{q-1} & \cdots & c_{q-1}^{q-2} & c_{q-1}^{q-1} \\
		\end{pmatrix}.
	\end{equation}
	Finally, we use them in a slightly modified version of Eq.~\ref{eq:toom-cook-full}:
	\begin{equation}
		xy = \mathbf{e}_{2^{n/k}}^{\intercal} (A'')^{-1} (A'\mathbf{x} \circ A'\mathbf{y}).
	\end{equation}
	
	Now for a couple of remarks. 
	In the previous paragraphs we have discussed only the standard Toom-Cook decomposition, but this strategy applies equally well to our modified Toom-Cook decomposition of $xyz$.
	In that case, the appropriate rows of $A'$ get scaled by $c_\ell^{k-1}$, and because there are three factors, this matches a factor of $c^{q-1} = c^{3(k-1)}$ in the corresponding rows of the $A''$ matrix.
	
	Finally, we note that while $A''$ has integer entries, that does not imply that $(A'')^{-1}$ will.
	Indeed, it has been shown that no choice of $w_\ell$ can avoid the need for division if $k>2$~\cite{bodrato_towards_2007}.
	For our algorithm this is not a problem at all---the entries of $(A'')^{-1}$ get wrapped up into the $\phi_\ell$ in classical precomputation, and thus the fact that they are not always integers does not affect the quantum circuit.
	This \textit{is} actually a big challenge for fast multiplication circuits other than ours, however, because it is necessary to perform division during the interpolation step, which can be expensive.
	This issue was encountered by a recent work which explored a quantum implementation of $k=3$ Toom-Cook, and found that the division operation led to a large Toffoli count~\cite{larasati_quantum_2021}.
	
	\section{Overflow bit phase application}
	\label{app:overflow-phase}

	As described in the main text, we avoid allocating extra registers to store the linear combinations $(A\mathbf{x})_\ell = \mathbf{x}^\intercal \mathbf{e}_{w_\ell}$ by performing addition \textit{in-place}, reusing existing qubits to store the summed values.\footnote{See footnote \ref{foot:overwrite}}
	However, because the value $(A\mathbf{x})_\ell$ may be a few bits larger than the portion of the input it is overwriting, without additional tricks this technique would yield a multiplication algorithm requiring $\mathcal{O}(1)$ ancillas at each level of the recursive tree.
	Here we describe a construction by which these extra ancillas can be avoided.

 \begin{figure*}
		\begin{center}
			\includegraphics[width=0.8\textwidth]{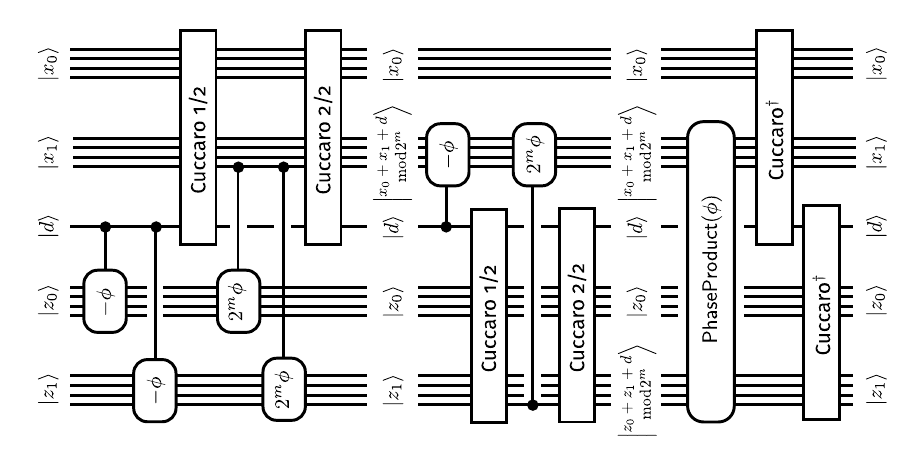}
		\end{center}
		\caption{A circuit for implementing a phase rotation by $\phi(x_0 + x_1)(z_0 + z_1)$, using only one dirty ancilla that is borrowed from another part of the computation.
			As discussed in the main text, this dirty ancilla can be borrowed from the inputs themselves, yielding an implementation that requires zero ancilla qubits, dirty or clean.
			Here, the unitaries with rounded corners represent phase rotations of the labeled factor times the value of the inputs on which the unitary is being applied.
			The operations labeled ``Cuccaro 1/2'' and ``Cuccaro 2/2'' represent the first and second halves of a Cuccaro adder, respectively.
			The unitaries controlled off of the top bit of a many-qubit register are controlled off of the qubit temporarily holding the value of the carry bit of the Cuccaro adder in progress.
			The operation ``Cuccaro$^\dag$'' is the Cuccaro adder performed backwards, uncomputing the sums.\label{fig:bit-shaving}}
	\end{figure*}
	
	\subsection{Simple case}

	For pedagogical purposes, let us first consider a simple example: the implementation of a \textsf{PhaseProduct} corresponding to a phase rotation of $\phi(x_0 + x_1)(z_0 + z_1)$, where $\phi$ is a classically-known phase factor, and we have as input the four registers $\ket{x_0}\ket{x_1}\ket{z_0}\ket{z_1}$.
	For simplicity in this first example, we will assume that all four registers are of the same length, which we can denote as $m$ bits.
	The straightforward method for implementing this phase rotation is to use a quantum adder to map $\ket{x_0}\ket{x_1} \to \ket{x_0}\ket{x_0 + x_1}$ (and similarly on the $z$ registers), and then use a recursive call to our \textsf{PhaseProduct} algorithm on the registers now holding $\ket{x_0 + x_1}$ and $\ket{z_0 + z_1}$.
	The issue is that the value $x_0 + x_1$ (resp. $z$) takes one more bit to represent than $x_1$, so naively we need to use an ancilla qubit to hold the ``carry bit'' produced at the top of the adder.
	This can be avoided as follows.
	
	Observe that we can split up the value $x_0 + x_1$ into the carry bit which we denote $x_c$ and the remaining $m$ bits of the sum, such that $x_0 + x_1 = 2^m x_c + (x_0 + x_1 \bmod 2^m)$ (and same for the $z$ values).
	Then with some rearrangement the product can be written as follows (the reason for the precise form of which will become clear promptly):
	\begin{widetext}
		\begin{equation}
			\phi (x_0 + x_1)(z_0 + z_1) =
			\phi 2^m x_c z_0 + \phi 2^m x_c z_1
			+ \phi 2^m z_c (x_0 + x_1 \bmod 2^m)
			+ \phi (x_0 + x_1 \bmod 2^m)(z_0 + z_1 \bmod 2^m)
			\label{eq:bit-shaving}
		\end{equation}
	\end{widetext}
	Thus we can implement our desired phase rotation by implementing separate rotations for each of the terms on the right hand side of Eq.~\ref{eq:bit-shaving}.
	
	Cuccaro's quantum adder~\cite{cuccaro_new_2004} has a variant that takes two length $m$ quantum integers, plus an ``incoming carry bit'' (of the same value as the least significant bits of the input), and computes the sum of the integers plus the incoming carry bit, writing the result into one of the input registers.
	It only requires one extra qubit, which is the qubit in which the \textit{outgoing} carry (e.g. $x_c$) is ultimately stored.
	Conveniently, the value of that outgoing carry is set simply with a CNOT from one of the other qubits, which temporarily holds the value $x_c$.
	This means that if we pause the Cuccaro adder mid-execution, we can forego the extra qubit needed to store the carry, as well as the CNOT setting it, instead using the value directly.
	Crucially, we have designed the first three terms of Eq.~\ref{eq:bit-shaving} such that they consist of the product of a single bit value by an $m$ bit value---and thus their phase rotation can be implemented directly via $m$ $CR_\phi$ gates, which is asymptotically negligible compared to the superlinear cost of the recursive call that is used to complete the last term in the sum.
	
	But what about the ``incoming carry bit'' into the Cuccaro adder?
	We don't have any incoming carry value here, but that qubit is crucial to the adder, so we cannot simply get rid of it (normally when there is no incoming carry, the Cuccaro adder uses an ancilla initialized to zero as the incoming carry).
	Instead we replace this incoming carry with a \textit{dirty qubit}---a qubit of unknown state, which we guarantee will be returned to its initial state after use.
	Dirty qubits are useful because they can be borrowed from idle parts of the computation, avoiding the allocation of an ancilla~\cite{barenco_elementary_1995, haner_factoring_2017, gidney_factoring_2018}.
	By linearity, if we can show that the correct operation is implemented with the dirty qubit in the $\ket{0}$ basis state and in the $\ket{1}$ basis state, the correct operation will be implemented using a dirty qubit in any state (even some complicated entangled one).
	The Cuccaro construction automatically returns the incoming carry to its original state, so we simply need to ensure that we can apply the correct phase rotation.
	The case in which the dirty qubit is in the $\ket{0}$ basis state is trivial---if the incoming carry is zero, the sum is unchanged.
	The tricky case is when the dirty qubit is in the $\ket{1}$ state, because the resulting sum will be one too large and then the phase rotation based on the resulting register will be slightly too large.
	We propose a simple fix: before starting the adder, perform a negative phase rotation controlled off of the dirty qubit, that exactly cancels the extra phase that will be applied due to the increase in the sum.
	This way, the total accumulated phase will be exactly the product that was originally desired.
	We show an explicit circuit for the application of the entire phase $\phi (x_0 + x_1)(z_0 + z_1)$ using only one, dirty ancilla in Figure~\ref{fig:bit-shaving}.

	Finally, we show how to avoid the need for even an external dirty ancilla qubit, making the multiplier entirely self-contained with zero ancillas.
	The idea is simple: stop the sums one bit earlier, not including the most significant bit of the inputs, and simply use one of these now-untouched bits as the dirty qubit input to the adder.
	Doing so will remove their value from the sum, but as was done with the carry bit, we can easily perform the portion of the \textsf{PhaseProduct} corresponding to those topmost bits by directly implementing the extra phase rotation via $m$ $CR_\phi$ gates.
	In fact, we are already performing $m$ $CR_\phi$ gates controlled off of the dirty ancilla to account for the extra phase it will add into the sum, so we can simply adjust those phase rotations to include the qubit's value when taken as the top bit of the inputs as well.
	
	\subsection{Full case}
	
	There are a few ways we must generalize the above to cover all situations that may arise in our algorithm: 1) the inputs may not be exactly the same length, for example if $n$ is not perfectly divisible by $k$; 2) there may be several terms in the sum, rather than simply two; 3) the terms of the sum may be multiplied by a constant; and 4) the \textsf{PhaseTripleProduct} has three factors instead of just two.
	It is easy to see how to handle the first three of these generalizations, by simply allowing more of the product to be done directly via $CR_\phi$ gates.
	As long as we ensure that the number of bits for which we perform the \textsf{PhaseProduct} directly is sufficiently small, it will not affect the asymptotic scaling.
	As shown in Figure~\ref{fig:bit-shaving}, we simply take the ``overlapping'' bits of the input and apply the algorithm from the previous sub-section to implement their \textsf{PhaseProduct}, directly using $\mathcal{O}(n)$ $CR_\phi$ gates to implement the phase rotation proportional to the other bits that do not overlap (a number of bits which is constant in $n$).
	We may use one of the non-overlapping bits as the dirty qubit input into the Cuccaro adder.
	For sums with multiple terms, we can simply apply the Cuccaro adder several times, making sure to rotate by an amount proportional to the carry bit of each sum when we have it available.
	By choosing the $w_\ell$ to be powers of 2 we ensure that any coefficients multiplying terms in our sums are also powers of two, and thus can be implemented as logical bit shifts in the inputs to the adders.
	
	\begin{figure}
		\includegraphics{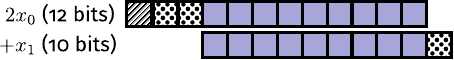}
		\caption{A diagram of the bits of two values, here a 12-bit value $x_0$ and a 10-bit value $x_1$ being summed to form the linear combination $2x_0 + x_1$.
		The setting is the implementation of a phase rotation proportional to, say, $(2x_0 + x_1)(2z_0 + z_1)$ ($z$ registers not shown).
		The purple bits are the ``main'' bits, which are summed together using Cuccaro's adder, with the top bit of $x_0$ (marked by diagonal hashes) being used as the ``incoming carry'' into Cuccaro's adder.
		Separately, a phase rotation is applied controlled off of that top bit of $x_0$, but adjusted for the fact that its value was carried into the sum of the main bits.
		The polka-dotted bits represent other extra ``non-overlapping'' bits whose phase rotations are implemented directly.
		Crucially, the number of extra bits is \textit{constant} in $n$, so the cost of performing the portion of the \textsf{PhaseProduct} involving them is asymptotically negligible.}
		\label{fig:main-extra-bits}
	\end{figure}
	
	The \textsf{PhaseTripleProduct} can be implemented without ancillas via essentially the same strategy, with one small hiccup: because there are three inputs, the portion of \textsf{PhaseTripleProduct} involving a constant number of qubits of one of the inputs requires $\mathcal{O}(1) \cdot \mathcal{O}(n) \cdot \mathcal{O}(n) = \mathcal{O}(n^2)$ gates to implement via the schoolbook algorithm of Eq.~\ref{eq:fourier-qq-schoolbook}, which is too many to maintain the larger algorithm's asymptotic scaling.
	Fortunately there is an easy workaround: a \textsf{PhaseTripleProduct} in which one of the inputs is a single qubit is equivalent to a \textsf{PhaseProduct} controlled off of that single bit.
	So, we may simply implement that portion of the operation via a constant number of controlled fast \textsf{PhaseProduct} operations, making sure to set $k$ in such a way that the asymptotic scaling of each \textsf{PhaseProduct} is better than the asymptotic scaling of the overall \textsf{PhaseTripleProduct} algorithm.
	Thus the contribution of these controlled \textsf{PhaseProduct} operations is asymptotically negligible and the overall asymptotic scaling of the fast \textsf{PhaseTripleProduct} is maintained.
	
	\section{Proof of circuit depth}
	\label{app:depth-proof}
	
	In this section of the Appendix, we prove Claim~\ref{claim:depth}, originally stated in Section~\ref{sec:spacetime} and reproduced here:
	
	\depthclaim*{}
	
	\begin{proof}
		The claim is equivalent to the statement that \textsf{PhaseProduct} and \textsf{PhaseTripleProduct} can be implemented in depth $\mathcal{O}(n^\epsilon)$ using $\mathcal{O}(n/\log n)$ ancilla qubits and $\mathcal{O}(n^{1+\epsilon})$ gates, since these operations conjugated by the QFT implement multiplication (and the ancilla qubits can be shared between the QFTs and the phase rotation circuits).
		Thus we proceed by proving that statement.
		For simplicity, we first give the proof for \textsf{PhaseProduct}; at the end we show that extending the proof to \textsf{PhaseTripleProduct} only requires trivial modifications.
		
		Let $\{w_\ell \}$ be the set of $q = 2k-1$ points at which the polynomials will be evaluated.
		Let $A_1$ be the $k \times k$ Vandermonde matrix formed from the first $k$ points, $w_0$ through $w_{k-1}$, and let $A_2$ be the $(k-1) \times k$ Vandermonde matrix for the $k-1$ points $w_{k}$ through $w_{q-1}$.
		The broad structure of the algorithm is as follows: first, the vectors $\bm{x}$ and $\bm{z}$ are each multiplied in-place by the matrix $A_1$, and then \textsf{PhaseProduct} is called in parallel on each of the $k$ elements of the result.
		Once those recursive calls are complete, the matrix-vector product with $A_1$ is inverted, and then the inputs are multiplied by $A_2$, such that \textsf{PhaseProduct} can again be applied in parallel, for the remaining $k-1$ $w_\ell$.
		Finally the matrix-vector product with $A_2$ is inverted, returning the vectors to their original values.
		
		The crux of the proof is to show that the matrix-vector multiplications can be performed in low depth without a large ancilla count.
		Consider a circuit $\mathcal{C}_A$ which implements the isometry $\ket{\bm{x}} \to \ket{A\bm{x}}\ket{g}$ for a classical integer matrix $A$ (which we will set to $A_1$ and $A_2$) where $\ket{g}$ is an ancillary ``garbage'' register that is allowed to be set to an arbitrary value by the circuit.
		Observe that this operation is sufficient for our needs: after the recursive calls to \textsf{PhaseProduct}, $\mathcal{C}_A$ is simply run backwards, uncomputing $g$ and returning the input register to $\ket{\bm{x}}$.
		For input register size $t$ qubits, let $D_{\mathcal{C}_A}(t)$ and $G_{\mathcal{C}_A}(t)$ be the depth and gate count, respectively, of $\mathcal{C}_A$, and let $S_{\mathcal{C}_A}(t)$ be the number of ``extra'' qubits used by $\mathcal{C}_A$ in addition to the $t$ input qubits.
		We use $S$ instead of $A$ to denote this extra qubit count, because unlike the ancillas discussed elsewhere in this manuscript, we do not require that these qubits are reset to zero by the end the circuit; in particular, they include $\ket{g}$ as well as the overflow bits used to accommodate the fact that $A\bm{x}$ is a few bits larger than $\bm{x}$.
		Then, the cost of the parallel algorithm for \textsf{PhaseProduct} described in the previous paragraph obeys the following recursion relations:
		\begin{itemize}
			\item Depth: $D(t) = 2D(t/k) + 2 D_{\mathcal{C}_{A_1}}(t) + 2 D_{\mathcal{C}_{A_2}}(t)$
			\item Ancillas: $A(t) = k A(t/k) + 2 \cdot \max(S_{\mathcal{C}_{A_1}}(t), S_{\mathcal{C}_{A_2}}(t))$
			\item Gates: $G(t) = (2k-1) G(t/k) + 4 G_{\mathcal{C}_{A_1}}(t) + 4 G_{\mathcal{C}_{A_2}}(t)$
		\end{itemize}
		In the next paragraphs we show how to construct $\mathcal{C}_{A_1}$ and $\mathcal{C}_{A_2}$ efficiently, in order to achieve the complexities stated in the Claim.
		
		The in-place matrix-vector product is straightforward for triangular matrices, by iterating through the matrix one row at a time (starting with the top row for an upper triangular matrix, and the bottom row for a lower triangular one).
		The matrices $A_1$ and $A_2$ are full rank, and thus can each be decomposed into the product of two triangular matrices via the LU decomposition (up to reordering of the rows), which should be precomputed classically.
		As shown in Appendix~\ref{app:fractional_w}, one may arbitrarily rescale rows of $A_1$ and $A_2$, which can be used to ensure that the LU decomposition consists entirely of integer values.
		Thus the quantum in-place matrix-vector product can thus be computed simply via a series of integer products, some in-place (corresponding to diagonal elements of the LU matrices) and some out-of-place (corresponding to off-diagonal elements).
		The total number of integer products is proportional to the size of the matrices $A_1$ and $A_2$, which is independent of $n$.
		For all of the integer products, one input is quantum and the other is an $O(1)$-bit classical value.
		
		Out-of-place integer multiplication by small classical constants is straightforward: to implement $\ket{x}\ket{y} \to \ket{x}\ket{y+cx}$ for an $\mathcal{O}(1)$-bit classical constant $c$, perform a shifted addition $\ket{x}\ket{y} \to \ket{x}\ket{y + 2^i x}$, for each bit $c_i$ of $c$ that is 1.
		Addition of quantum values can be performed in $\mathcal{O}(\log n)$ depth using $\mathcal{O}(n/\log n)$ ancilla qubits (which are released after the operation) and $\mathcal{O}(n)$ total gates~\cite{takahashi_fast_2008}; the out-of-place multiply-add can be performed in a constant number of these additions since $c$ has $O(1)$ bits.
		
		In-place integer multiplication is less straightforward, and seems to require the use of the output garbarge register $\ket{g}$ to be performed in polylogarithmic depth.
		We show in Lemma~\ref{lem:in-place-constant-mult} that for a classical constant $c$ whose bit length is independent of $t$ (the length of the quantum input), the isometry $\ket{x} \to \ket{cx}\ket{g}$ can be implemented with depth $\mathcal{O}(\operatorname{polylog}(t))$, extra space $\mathcal{O}(t/\log^2 t)$, and gate count $\mathcal{O}(t)$.
		
		Thus, we find that overall, the matrix-vector multiplications can be performed in $D_{\mathcal{C}_A}(t) = \mathcal{O}(\operatorname{polylog}(t))$, $S_{\mathcal{C}_A}(t) = \mathcal{O}(t/\log^2 t)$, and $G_{\mathcal{C}_A}(t) = \mathcal{O}(t)$.
		Using these values in the recursion relations for \textsf{PhaseProduct} yields the complexities:
		\begin{itemize}
			\item Depth: $D(n) = \mathcal{O}(n^{\log_k 2})$
			\item Ancillas: $A(n) = \mathcal{O}(n/\log n)$
			\item Gates: $G(n) = \mathcal{O}(n^{\log_k (2k-1)})$
		\end{itemize}
		The choice of when to stop the recursion is important to achieve this ancilla qubit count.
		If we were to recurse to a base case of size $O(1)$, each of the $O(n)$ leaves of the recursive tree would require $O(1)$ ancillas.
		Instead we observe that a base case of size $t_b$ can be implemented via the schoolbook algorithm in depth $O(t_b)$ using zero ancillas; thus, we can switch to the base case at a size of $t_b = \mathcal{O}(n^{\log_k 2})$ without adversely affecting the depth, which maintains the sublinear qubit count. 
		
		The same proof applies to \textsf{PhaseTripleProduct}, except that $q=3k-2$, and we thus must apply the recursive calls in three layers, instead of two.
		This yields the following complexities for \textsf{PhaseTripleProduct}:
		\begin{itemize}
			\item Depth: $D(n) = \mathcal{O}(n^{\log_k 3})$
			\item Ancillas: $A(n) = \mathcal{O}(n/\log n)$
			\item Gates: $G(n) = \mathcal{O}(n^{\log_k (3k-2)})$
		\end{itemize}
		
		For both \textsf{PhaseProduct} and \textsf{PhaseTripleProduct}, for any $\epsilon > 0$ there exists a value of the constant $k$ which yields depth smaller than $O(n^\epsilon)$ and gate count smaller than $O(n^{1+\epsilon})$, completing the proof.
	\end{proof}
	
	\begin{lemma}\label{lem:in-place-constant-mult}
		Given an $n$-qubit register $\ket{x}$ and $O(1)$-bit classical constant $c$, the isometry $\ket{x} \to \ket{cx}\ket{g(x)}$, where $g$ is unspecified ``garbage'' data, can be performed in depth $\mathcal{O}(\operatorname{polylog}(n))$ and gate count $O(n)$ using a total of $n + O(n/\log^2 n)$ qubits including the input register.
	\end{lemma}
	
	\begin{proof}
		The folowing algorithm achieves the desired operation:
		\begin{enumerate}
			\item Let $d$ be the integer such that $c' = c/2^d$ is an odd integer.
			\item Divide $\ket{x}$ into blocks of $m = O(\log^2 n)$ qubits each. Label the value of block $i$ $\ket{X_i}$.
			\item Locally perform the mapping $\ket{X_i} \to \ket{c' X_i \bmod 2^m}\ket{B_i}$ on each block, where $B_i = \lfloor c' X_i / 2^m \rfloor$ is the value of the $\mathcal{O}(1)$ ``overflow'' bits, stored in ancilla qubits, and $\ket{c' X_i \bmod 2^m}$ overwrites block $i$ of the input register.
			\item Increment the entire input register by $b = \sum_i 2^{mi} B_i$.
			\item Multiply the input register by $2^d$ (via a bit shift).
		\end{enumerate}
		
		We now briefly explain each step, and how it can be performed efficiently.
		Steps 1 and 2 are entirely classical.
		In step 3, the mapping of $\ket{X_i} \to \ket{c'X_i}$ on each block (here not distinguishing between overflow and input qubits) can be implemented in depth $O(\operatorname{polylog} n)$ using $O(1)$ ancillas per block, by performing an in-place addition of the value $2^j(c'-1)$ controlled off of each bit $j$ of $X_i$.
		By iterating from the most significant bit of $X_i$ to the least significant, it is ensured that later control bits are not affected by previous additions.
		Here we may use standard linear-depth adders, because the size of the additions is the block size $m = \mathcal{O}(\log^2 n)$.
		At the end of step 3, the full state can be written $\ket{c'x-b}\ket{b}$ for $b = \sum_i 2^{mi} B_i$ as above, where here the first register is the input register, and the second is a register of $O(n/\log^2 n)$ ancilla qubits representing $b$ in a sparse format of $O(1)$ qubits per $m$-qubit block (the rest of the bits of $b$ are zero).
		Step 4 then maps the state to $\ket{c'x}\ket{b}$; this mapping can be performed in depth $\mathcal{O}(\log n)$ using $O(1)$ ancillas per block (that are returned to $\ket{0}$ immediately) via a low-space quantum carry-lookahead adder construction~\cite{takahashi_fast_2008}, modified slightly to account for the fact that $b$ is sparse.
		Finally, the bitshift of step 5 by a classical constant $d$ yields a final state $\ket{cx}\ket{b}$, which is the desired final state (with $g(x) = b$).
	\end{proof}
	
	\section{In-place modular multiplication}
	\label{app:shor}
	
	The core of Shor's algorithm for factoring an integer of $n$ bits can be implemented via a series of in-place multiplications by classical constants, controlled off of a single qubit (see the ``one controlling qubit'' trick, Sec. 2.4 of~\cite{beauregard_circuit_2003}).
	Written out, a single one of these multiplications consists of the in-place operation
	\begin{equation}
		\ket{x} \to \ket{c x \bmod N}
		\label{eq:inplace-mult}
	\end{equation}
	for a classical integer $c$.
	Here we apply the modular multiplication introduced in Sec.~\ref{sec:applications} to implement this operation to within an error $\eta$ using $2n + \mathcal{O}(\log 1/\eta)$ qubits (while maintaining the same subquadratic gate count of the multiplication algorithm).
	
	To do so, we make the following observation: our classical-quantum multiplication algorithm does not require that the classical value $c$ be an integer.
	It can be a floating-point number, which we may classically compute to whatever arbitrary precision we desire before using it to compute the values of the phase rotations in our multiplication algorithm.
	We use this fact to perform roughly the following trick: compute the fractional value $w = (cx \bmod N)/N$ up to some precision $m$, and then multiply by $N$ to convert the fractional value $w$ into an integer.
	The accuracy of this operation will depend on the precision to which we compute $w$; however, we will find that we only need $\mathcal{O}(\log(1/\eta))$ extra bits to achieve an error of less than $\eta$.
	
	Algorithm~\ref{alg:in-place} applies this idea to approximately implement the unitary of Eq.~\ref{eq:inplace-mult}.
	It is clear by inspection that it uses $2n + \mathcal{O}(\log 1/\eta)$ total qubits and $\mathcal{O}(n^{\log_k q})$ gates.
	In the following theorem we prove the error bound.
	
	\begin{algorithm2e}
		\caption{In-place classical-quantum modular multiplication}
		\label{alg:in-place}
		\vspace{5pt}
		\SetKwInOut{Input}{Input}
		\Input{%
			\ Quantum state $\ket{x}$ (extended to superpositions by linearity) \\
			\ Classical constant $c$ \\
			\ Error level $\eta$
		}
		\SetKwInOut{Output}{Output}
		\Output{%
			\ Quantum state $\ket{cx \bmod N}$ (up to error $\eta$)
		}
		\vspace{5pt}
		Let $m = n + \lceil 2 \log (2 + 1/2\eta) \rceil$ \\
		Allocate a register of $m$ ancillas initialized to $\ket{0}$ \\
		\vspace{5pt}
		\nl Compute $\ket{x}\ket{0} \to \ket{x}\ket{w}$ for $w = ((c-1)x \mod N) / N$ via classical-quantum multiplication \\
		\nl Adding one ancilla to the top of the first register, compute $\ket{x}\ket{w} \to \ket{x + Nw}\ket{w}$. State is now (approximately) $\ket{cx \bmod N}\ket{w}$ or $\ket{(cx \bmod N) + N}\ket{w}$. \\
		\nl Using an ancilla qubit, compute whether the left register is greater than $N$; subtract $N$ controlled by the ancilla. State is now $\ket{cx \bmod N}\ket{w}$ \\
		\nl Uncompute the ancilla qubit by computing whether $cx \bmod N < Nw $ via a comparison operator \\
		\nl Subtract the value $w = ((1 - c^{-1})(cx) \bmod N) / N$ from the second register, where $c^{-1}$ is the multiplicative inverse of $c \pmod N$. State is now $\ket{cx \bmod N}\ket{0}$.
	\end{algorithm2e}
	
	\begin{thm}
		The final state $\ket{\psi}$ produced by Algorithm~\ref{alg:in-place} has $|1 - \braket{\psi | cx \bmod N}| < \mathcal{O}(\eta)$, for arbitrary inputs $\ket{x}$ and $c$.
	\end{thm}
	\begin{proof}
		We enumerate each step of the algorithm, at various points replacing the state at that step with another state that is $\eta$-close (measured via inner product as above).
		By the triangle inequality and the fact that unitary transformations preserve inner products, we thus prove the theorem.
		
		We begin with step 1.
		Applying Hadamard gates to the second register to generate $\ket{x}\sum_z \ket{z}$, then a phase rotation of $\phi = 2\pi wz$ and then an inverse $\mathsf{QFT}$ modulo $2^m$ yields a state which in full generality can be written
		\begin{equation}
			\ket{x} \sum_{\tilde{w}} \alpha_{\tilde{w}} \ket{\tilde{w}}
		\end{equation}
		where the sum is over all $m$-bit binary fractions $\tilde{w}$.
		Let $S$ be the set of $m$-bit binary fractions for which $|w-\tilde{w}| < \eta/2^n$.
		The key to Algorithm~\ref{alg:in-place} is that from quantum phase estimation (see Sec.~5.2.1 of~\cite{nielsen_quantum_2011}), we have that $\sum_{\tilde{w} \in S} \alpha_{\tilde{w}} \ket{\tilde{w}} \approx \sum_{\tilde{w}} \alpha_{\tilde{w}} \ket{\tilde{w}}$; formally, the inner product of the actual state and the truncated state in which only terms with $|w-\tilde{w}| < \eta/2^n$ are kept is within $\mathcal{O}(\eta)$ of 1.
		Therefore, we may consider for the analysis of step 2 the truncated state for which $|w-\tilde{w}| < \eta/2^n$ for all terms in the superposition.
		
		Applying a quantum Fourier transform now to the register containing $\ket{x}$, and then applying $\mathsf{PhaseProduct}(2\pi 2^{2n-m}/N)$ to the two registers (where the second register's value is now considered as an integer) yields the state
		\begin{equation}
			\sum_z e^{2\pi i (x+((c-1)x \bmod N)z/2^{n+1}} \ket{z} \sum_{\tilde{w}} \alpha_{\tilde{w}} \ket{\tilde{w}}
		\end{equation}
		up to a phase error of at most $\mathcal{O}(\eta)$.
		Replacing that state with the state having precisely the correct phase, and applying an inverse quantum Fourier transform to the first register yields
		\begin{equation}
			\ket{x+((c-1)x \bmod N)} \sum_{\tilde{w}} \alpha_{\tilde{w}} \ket{\tilde{w}}.
		\end{equation}
		The conditional subtraction of steps 3 and 4 is exact, yielding the state
		\begin{equation}
			\ket{cx \bmod N)} \sum_{\tilde{w}} \alpha_{\tilde{w}} \ket{\tilde{w}}.
		\end{equation}
		Now performing step 5 by precisely inverting the operations of step 1 that produced the state of the second register yields the all zero state on the second register, completing the computation.
	\end{proof}
	
	For Shor's algorithm, this operation must be repeated $\mathcal{O}(n)$ times.
	Thus for each product we set $\eta = \mathcal{O}(1/n)$ so that the error of the overall algorithm does not grow with $n$.
	This yields a total qubit count of $2n + \mathcal{O}(\log n/\eta')$, where $\eta'$ is the error over the entire algorithm.

 	\section{Estimation of gate and qubit counts}
    \label{app:estimation_details}
    
    In this section, we describe the details of the resource estimates presented in Tables~\ref{tab:cq-2048-costs} and \ref{tab:x2modN-costs} of the main text.
    The resource counts for previous works were computed using code presented with those works; in this section we describe the estimation for the algorithms presented in this work. 
    The code used to perform these estimates is available online: \url{https://zenodo.org/doi/10.5281/zenodo.10871109}.
    
    \subsection{Recursive structure and choice of $k$}
    
    One of the most critical choices for efficiency is the choice of $k$, the number of pieces into which we divide the values being multiplied.
    Because increasing $k$ decreases the asymptotic cost but increases costant factors, the obvious strategy is to find some ``cutoff'' value for each $k$, and use the largest $k$ for which the desired product is above the cutoff.
    However, the situation is complicated by the fact that in practice, the cost for a particular $k$ does not increase smoothly with $n$---there is a range of ``overlap'' where neighboring $k$ alternately jump past each other in efficiency. 
  	(This is due to the fact that $k$ in general will not evenly divide $n$).
  	It is complicated even further by consideration that the linear combinations like $(Ax)_\ell$ may be slightly different bit lengths for different $\ell$, so it may be optimal to use different $k$ for different $\ell$, even at the same level of recursion.
    
    Because of all of those complications, for each recursive call we simply exhaustively search through all $k$, and choose the construction that yields the lowest cost (for some cost function that we may define; we use the number of $CR_\phi$ gates as our cost function).
    In practice, we can perform this exhaustive search in a reasonable amount of time because the tree depth is logarithmic in $n$, and we memoize the optimal result for each value of $n$ encountered in the recursion.
    
    \subsection{Evaluation points $w_\ell$}
    
	We choose the $q$ evaluation points $w_\ell$ as follows.
	We always include $0$ and $\infty$ because both can be evaluated directly from the inputs, without the need for any linear combinations.
	Next we include $-1$ and $1$ if needed, prioritizing $-1$ because the absolute value of the resulting linear combination is bounded to a smaller value.
	For $q>4$, we then use the points $-1/\omega, 1/\omega, -\omega, \omega$ in that order, for $\omega$ powers of 2.
	We use powers of 2 because then the resulting coefficients are powers of 2 and thus values can be scaled by them via a logical bit shift.
    We prioritize unit fractions over whole numbers because it helps reduce the size of the linear combinations: the most-significant chunk (e.g. $x_{k-1}$) will in general be a few bits longer than all the rest, because $n$ will in general not be evenly divisble by $k$.
    Thus we would like to scale it by the smallest coefficient, which occurs with unit fraction $w_\ell$ (see Appendix~\ref{app:fractional_w}).
    We prioritize negative coefficients for the same reason as above with $-1$, that it helps reduce the absolute value of the linear combination which reduces the ancilla counts.
    
    For pairs of $w_\ell$ that are the negative of each other, every negative term is summed together first, for a value we may denote $x_1'$.
    Then $x_{-w_\ell} = x_{w_\ell} - 2x_1'$, requiring a single addition/subtraction.
    Using this strategy, computing and uncomputing both $x_{w_\ell}$ for a pair $\pm w_\ell$ requires a total of $2k-1$ additions, while computing and uncomputing a single unpaired $w_\ell$ requires $2k-2$ additions.
    This strategy requires a few extra ancilla qubits, to hold the overflow of both the intermediate sum register and the register holding the final sum result.
    
    In all cases, the register we choose to overwrite with the linear combinations is the most significant (e.g. $x_{k-1}$), because it may be a few bits longer than the others and thus require fewer ancilla qubits to store the value.
    
    Next, we describe certain optimizations that are specific to the individual applications described in the main text.

    \subsection{2048-bit classical-quantum multiplication}
    
    For non-modular arithmetic, the output register for this operation is twice as long as the input.
    Thus, in that case we implement the whole \textsf{PhaseProduct} as two simultaneous \textsf{PhaseProduct} operations, corresponding to the phase rotations for the low and high halves of the output $z$ register respectively.
    We perform both simultaneously so that we only have to compute the linear combinations on the input $x$ register once rather than twice.
    This uses a small amount more ancillas because the same ancillas cannot be reused for each of the two halves of the output register; we consider this tradeoff worth the reduction in gate counts.
    
    The base case we use in this construction is very simple: we simply apply $n_\mathrm{base}^2$ $CR_\phi$ gates, as in Eq.~\ref{eq:fourier-schoolbook}.
    As discussed in Sec.~\ref{sec:applications} of the main text, it is possible to reduce the number of $CR_\phi$ gates further by performing the base case product explicitly into an ancilla register and applying $R_\phi$ gates on the result.
    
    \subsection{$x^2 \bmod N$ proof of quantumness}
    
    A powerful feature of the proof of quantumness protocol is that ``garbage bits'' can be discarded for free via measurement~\cite{kahanamoku-meyer_classically_2022}.
    We use this fact to reduce gate counts at the expense of a moderate number of ancillas, by storing the $(A\mathbf{x})_i$ etc. in separate ancilla registers and then measuring them away instead of using adder circuits to uncompute them.
    This corresponds to the ``fast'' version of our circuits, listed in Table~\ref{tab:x2modN-costs}.
    The ``balanced'' version computes and uncomputes the sums in-place, requiring a few more gates but many fewer ancillas.
    
    The ``narrow'' version uses the decomposition $x^2z = z_0(x_0^2 + x_0 x_1 + x_1^2) + z_1(x_0^2 + x_0 x_1 + x_1^2)$ to reduce qubit counts quite a bit further, at the expense of some more gates.
    In particular, it first computes the low half of the bits of the output $x^2 \bmod N$ via the first term above (that proportional to $z_0$), performing half of the final QFT and measuring the result.
    Then, the same $n/2$ qubits can be reused to compute the high half of the output bits and the remaining half of the final QFT.
    This yields a total qubit count of roughly only $3n/2$ qubits because all of the bits of the $n$-bit long output never need to be stored at the same time.
    For the products themselves, we compute the additions in-place as in the ``balanced'' construction to keep the qubit counts low.
    (It is possible to take decompositions like this further to reduce qubit counts even more, but in preliminary exploration we find that the increase in gate counts is not worth the trade off).
    
    For this construction, we use a more involved base case than in the classical-quantum case.
    Notationally, denote the phase rotation we desire to implement in the base case as $\phi' x'^2 z'$.
    Broadly, as described in Sec.~\ref{sec:applications} of the main text, the idea is to replace the $n_\mathrm{base}^3$ $CCR_\phi$ gates that would be required to implement the schoolbook \textsf{PhaseTripleProduct} of Eq.~\ref{eq:fourier-qq-schoolbook} with $2n_\mathrm{base}^2$ $CR_\phi$ gates, by explicitly computing the product $x'^2$ and then performing $CR_\phi$ gates between the bits of that product and the bits of $z'$.
    However, because of the ability to freely uncompute garbage bits, we can avoid having to allocate $2n_\mathrm{base}$ qubits to store $x'^2$ in its entirety.
    Instead, we decompose $x'^2 = \sum_i 2^i \sum_j x'_j x'_{i-j}$---that is, we group the bit products by the power of 2 they are scaled by.
    Following this grouping, we first compute the sum for $i=0$, apply the $CR_\phi$ gates that involve the least significant bit of the result, and then immediately measure it away to uncompute it.
    Next we compute $i=1$, summing into any existing carry bits from $i=0$, and again perform controlled rotations off of the least significant bit before measuring it away.
    We continue this process for all $i$.
    Throughout this process, we will need to store at most $\log_2 n_\mathrm{base}$ carry bits, a substantial improvement over the $2n_\mathrm{base}$ ancillas that would be naively required.
    
    Finally, we note that in the quantum-quantum case, the $k=2$ algorithm scales as $\mathcal{O}(n^2)$ just like the base case, but is observed to have worse constants.
    Therefore we always move directly to the base case from $k=3$, never using the $k=2$ construction.
	
	\section{Parallel sequences}
	\label{app:parallel-sequences}
	
	In Tables~\ref{tab:k3_cq_ops}, \ref{tab:k3_qq_ops}, and \ref{tab:k4_qq_ops}, we provide examples of sequences of operations that allow for the parallel computation of branches of the recursive tree, reducing the overall depth of the circuit.
	
	\begin{table*}
		\begin{center}
			\begin{tabular}{|c|c|c|c|c|c|c|}
				\hline
				Operation & Register 0 & Register 1 & Register 2 \\
				\hline
				(start) & $\ket{x_0}$  & $\ket{x_1}$ & $\ket{x_2}$ \\
				Add reg. 2 to reg. 1 & $\ket{x_0}$  & $\ket{x_1 + x_2}$ & $\ket{x_2}$ \\
				Add reg. 0 to reg. 1 & $\ket{x_0}$  & $\ket{x_0 + x_1 + x_2}$ & $\ket{x_2}$ \\
				\textbf{Product on all registers} & $\mathbf{\ket{x_0}}$  & $\mathbf{\ket{x_0 + x_1 + x_2}}$ & $\mathbf{\ket{x_2}}$ \\
				Invert sign of reg. 1 & $\ket{x_0}$  & $\ket{-x_0 - x_1 - x_2}$ & $\ket{x_2}$ \\
				Add reg. 0 to reg. 1 & $\ket{x_0}$  & $\ket{- x_1 - x_2}$ & $\ket{x_2}$ \\
				Add $2\times$ reg. 2 to reg. 1 & $\ket{x_0}$  & $\ket{- x_1 + x_2}$ & $\ket{x_2}$ \\
				Add reg. 1 to reg. 0 & $\ket{x_0 - x_1 + x_2}$  & $\ket{- x_1 + x_2}$ & $\ket{x_2}$ \\
				Add reg. 0 to reg. 1 & $\ket{x_0 - x_1 + x_2}$  & $\ket{x_0 - 2x_1 + 2x_2}$ & $\ket{x_2}$ \\
				Add $2\times$ reg. 2 to reg. 1 & $\ket{x_0 - x_1 + x_2}$  & $\ket{x_0 - 2x_1 + 4x_2}$ & $\ket{x_2}$ \\
				\textbf{Product on regs. 1 and 0} & $\mathbf{\ket{x_0 - x_1 + x_2}}$  & $\mathbf{\ket{x_0 - 2x_1 + 4x_2}}$ & $\ket{x_2}$ \\
				Invert sign of reg. 1 & $\ket{x_0 - x_1 + x_2}$  & $\ket{- x_0 + 2x_1 - 4x_2}$ & $\ket{x_2}$ \\
				Add $2\times$ reg. 2 to reg. 1 & $\ket{x_0 - x_1 + x_2}$  & $\ket{-x_0 +2x_1 -2x_2}$ & $\ket{x_2}$ \\
				Add reg. 1 to $2\times$ reg. 0 & $\ket{x_0}$  & $\ket{-x_0 +2x_1 -2x_2}$ & $\ket{x_2}$ \\
				Add reg. 0 to reg. 1 & $\ket{x_0}$  & $\ket{2x_1 - 2x_2}$ & $\ket{x_2}$ \\
				Add $2\times$ reg. 2 to reg. 1 & $\ket{x_0}$  & $\ket{2x_1}$ & $\ket{x_2}$ \\
				Divide reg. 1 by two & $\ket{x_0}$  & $\ket{x_1}$ & $\ket{x_2}$ \\
				\hline
			\end{tabular}
		\end{center}
		\caption{
			\label{tab:k3_cq_ops}
			\textbf{$\mathbf{k=3}$ parallel sequence for \textsf{PhaseProduct}.}
			This table lists the quantum operations performed to implement the unitary $\tilde{\mathcal{U}}_{q\times c}(a)$.
			The registers are divided into subregisters as $\ket{x} = \ket{x_0}\ket{x_1}\ket{x_2}$ and $\ket{z} = \ket{z_0}\ket{z_1} \ket{z_2}$ (using little-endian notation, so $x_0$ is the least-significant subregister).
			In this table only the state of the $x$ sub-registers are shown; the same operations are applied to the $z$ register.
			``Product on registers'' means to apply a phase corresponding to the product of the respective $x$ and $z$ registers, usually by recursively calling the same algorithm again.
			Registers containing values upon which the algorithm is applied recursively are highlighted in bold.
			The linear combinations used here for the products correspond to the evaluation points $w_\ell \in \{0, \infty, \pm 1, -2\}$.
		}
	\end{table*}
	
	\begin{table*}
		\begin{center}
			\begin{tabular}{|c|c|c|c|}
				\hline
				Operation & Register 0 & Register 1 & Register 2 \\
				\hline
				(start) & $\ket{x_0}$ & $\ket{x_1}$ & $\ket{x_2}$\\
				Add  reg. 0 to  reg. 1 & $\ket{x_0}$ & $\ket{x_0+x_1}$ & $\ket{x_2}$\\
				Add  reg. 2 to  reg. 1 & $\ket{x_0}$ & $\ket{x_0+x_1+x_2}$ & $\ket{x_2}$\\
				\textbf{Product on all} & $\mathbf{\ket{x_0}}$ & $\mathbf{\ket{x_0+x_1+x_2}}$ & $\mathbf{\ket{x_2}}$\\
				Add $-1 \times$ reg. 2 to  reg. 1 & $\ket{x_0}$ & $\ket{x_0+x_1}$ & $\ket{x_2}$\\
				Add $2 \times$ reg. 0 to $-1 \times$ reg. 1 & $\ket{x_0}$ & $\ket{x_0-x_1}$ & $\ket{x_2}$\\
				Add  reg. 1 to  reg. 2 & $\ket{x_0}$ & $\ket{x_0-x_1}$ & $\ket{x_0-x_1+x_2}$\\
				Add  reg. 2 to  reg. 1 & $\ket{x_0}$ & $\ket{2x_0-2x_1+x_2}$ & $\ket{x_0-x_1+x_2}$\\
				Add $2 \times$ reg. 0 to  reg. 1 & $\ket{x_0}$ & $\ket{4x_0-2x_1+x_2}$ & $\ket{x_0-x_1+x_2}$\\
				\textbf{Product on regs. 1 and 2} & $\ket{x_0}$ & $\mathbf{\ket{4x_0-2x_1+x_2}}$ & $\mathbf{\ket{x_0-x_1+x_2}}$\\
				Add  reg. 0 to $2 \times$ reg. 2 & $\ket{x_0}$ & $\ket{4x_0-2x_1+x_2}$ & $\ket{3x_0-2x_1+2x_2}$\\
				Add  reg. 2 to $-2 \times$ reg. 1 & $\ket{x_0}$ & $\ket{-5x_0+2x_1}$ & $\ket{3x_0-2x_1+2x_2}$\\
				Add $2 \times$ reg. 2 to  reg. 1 & $\ket{x_0}$ & $\ket{x_0-2x_1+4x_2}$ & $\ket{3x_0-2x_1+2x_2}$\\
				Add  reg. 1 to $-4 \times$ reg. 2 & $\ket{x_0}$ & $\ket{x_0-2x_1+4x_2}$ & $\ket{-11x_0+6x_1-4x_2}$\\
				Add $2 \times$ reg. 1 to  reg. 2 & $\ket{x_0}$ & $\ket{x_0-2x_1+4x_2}$ & $\ket{-9x_0+2x_1+4x_2}$\\
				Add $8 \times$ reg. 0 to  reg. 2 & $\ket{x_0}$ & $\ket{x_0-2x_1+4x_2}$ & $\ket{-x_0+2x_1+4x_2}$\\
				Add $2 \times$ reg. 0 to  reg. 2 & $\ket{x_0}$ & $\ket{x_0-2x_1+4x_2}$ & $\ket{x_0+2x_1+4x_2}$\\
				\textbf{Product on regs. 1 and 2} & $\ket{x_0}$ & $\mathbf{\ket{x_0-2x_1+4x_2}}$ & $\mathbf{\ket{x_0+2x_1+4x_2}}$\\
				Add  reg. 2 to $-1 \times$ reg. 1 & $\ket{x_0}$ & $\ket{4x_1}$ & $\ket{x_0+2x_1+4x_2}$\\
				Add $-1 \times$ reg. 0 to  reg. 2 & $\ket{x_0}$ & $\ket{4x_1}$ & $\ket{2x_1+4x_2}$\\
				Add $-1 \times$ reg. 1 to $2 \times$ reg. 2 & $\ket{x_0}$ & $\ket{4x_1}$ & $\ket{8x_2}$\\
				Divide reg. 1 by 4 & $\ket{x_0}$ & $\ket{x_1}$ & $\ket{8x_2}$\\
				Divide reg. 2 by 8 & $\ket{x_0}$ & $\ket{x_1}$ & $\ket{x_2}$\\
				\hline
			\end{tabular}
		\end{center}
		\caption{
			\label{tab:k3_qq_ops}
			\textbf{$\mathbf{k=3}$ parallel sequence for \textsf{PhaseTripleProduct}.}
			In this table only the state of the $x$ sub-registers are shown; the same operations are applied to the $y$ and $z$ registers.
			Registers containing values upon which the algorithm is applied again recursively are highlighted in bold.
			The linear combinations used here for the products correspond to the evaluation points $w_\ell \in \{0, \infty, \pm 1, \pm 2, -1/2 \}$.
		}
	\end{table*}
	
	\begin{table*}
		\begin{center}
				\begin{tabular}{|c|c|c|c|c|}
					\hline
					Operation & Register 0 & Register 1 & Register 2 & Register 3 \\
					\hline
					(start) & $\ket{x_0}$ & $\ket{x_1}$ & $\ket{x_2}$ & $\ket{x_3}$\\
					Add  reg. 0 to  reg. 2 & $\ket{x_0}$ & $\ket{x_1}$ & $\ket{x_0+x_2}$ & $\ket{x_3}$\\
					Add  reg. 3 to  reg. 1 & $\ket{x_0}$ & $\ket{x_1+x_3}$ & $\ket{x_0+x_2}$ & $\ket{x_3}$\\
					Add  reg. 2 to  reg. 1 & $\ket{x_0}$ & $\ket{x_0+x_1+x_2+x_3}$ & $\ket{x_0+x_2}$ & $\ket{x_3}$\\
					Add  reg. 1 to $-2 \times$ reg. 2 & $\ket{x_0}$ & $\ket{x_0+x_1+x_2+x_3}$ & $\ket{-x_0+x_1-x_2+x_3}$ & $\ket{x_3}$\\
					\textbf{Product on all regs.} & $\mathbf{\ket{x_0}}$ & $\mathbf{\ket{x_0+x_1+x_2+x_3}}$ & $\mathbf{\ket{-x_0+x_1-x_2+x_3}}$ & $\mathbf{\ket{x_3}}$\\
					Add  reg. 1 to  reg. 2 & $\ket{x_0}$ & $\ket{x_0+x_1+x_2+x_3}$ & $\ket{2x_1+2x_3}$ & $\ket{x_3}$\\
					Add $-3 \times$ reg. 3 to  reg. 2 & $\ket{x_0}$ & $\ket{x_0+x_1+x_2+x_3}$ & $\ket{2x_1-x_3}$ & $\ket{x_3}$\\
					Add $3 \times$ reg. 0 to  reg. 1 & $\ket{x_0}$ & $\ket{4x_0+x_1+x_2+x_3}$ & $\ket{2x_1-x_3}$ & $\ket{x_3}$\\
					Add  reg. 2 to $2 \times$ reg. 1 & $\ket{x_0}$ & $\ket{8x_0+4x_1+2x_2+x_3}$ & $\ket{2x_1-x_3}$ & $\ket{x_3}$\\
					Add $3 \times$ reg. 3 to $2 \times$ reg. 2 & $\ket{x_0}$ & $\ket{8x_0+4x_1+2x_2+x_3}$ & $\ket{4x_1+x_3}$ & $\ket{x_3}$\\
					Add $-1 \times$ reg. 1 to $2 \times$ reg. 2 & $\ket{x_0}$ & $\ket{8x_0+4x_1+2x_2+x_3}$ & $\ket{-8x_0+4x_1-2x_2+x_3}$ & $\ket{x_3}$\\
					\textbf{Product on regs. 1 and 2} & $\ket{x_0}$ & $\mathbf{\ket{8x_0+4x_1+2x_2+x_3}}$ & $\mathbf{\ket{-8x_0+4x_1-2x_2+x_3}}$ & $\ket{x_3}$\\
					Add  reg. 1 to $-1 \times$ reg. 2 & $\ket{x_0}$ & $\ket{8x_0+4x_1+2x_2+x_3}$ & $\ket{16x_0+4x_2}$ & $\ket{x_3}$\\
					Divide reg. 2 by 4 & $\ket{x_0}$ & $\ket{8x_0+4x_1+2x_2+x_3}$ & $\ket{4x_0+x_2}$ & $\ket{x_3}$\\
					Add $6 \times$ reg. 2 to  reg. 1 & $\ket{x_0}$ & $\ket{32x_0+4x_1+8x_2+x_3}$ & $\ket{4x_0+x_2}$ & $\ket{x_3}$\\
					Add $-15 \times$ reg. 0 to $4 \times$ reg. 2 & $\ket{x_0}$ & $\ket{32x_0+4x_1+8x_2+x_3}$ & $\ket{x_0+4x_2}$ & $\ket{x_3}$\\
					Add $15 \times$ reg. 3 to  reg. 1 & $\ket{x_0}$ & $\ket{32x_0+4x_1+8x_2+16x_3}$ & $\ket{x_0+4x_2}$ & $\ket{x_3}$\\
					Divide reg. 1 by 2 & $\ket{x_0}$ & $\ket{16x_0+2x_1+4x_2+8x_3}$ & $\ket{x_0+4x_2}$ & $\ket{x_3}$\\
					Add $-15 \times$ reg. 0 to  reg. 1 & $\ket{x_0}$ & $\ket{x_0+2x_1+4x_2+8x_3}$ & $\ket{x_0+4x_2}$ & $\ket{x_3}$\\
					Add  reg. 1 to $-2 \times$ reg. 2 & $\ket{x_0}$ & $\ket{x_0+2x_1+4x_2+8x_3}$ & $\ket{-x_0+2x_1-4x_2+8x_3}$ & $\ket{x_3}$\\
					\textbf{Product on regs. 1 and 2} & $\ket{x_0}$ & $\mathbf{\ket{x_0+2x_1+4x_2+8x_3}}$ & $\mathbf{\ket{-x_0+2x_1-4x_2+8x_3}}$ & $\ket{x_3}$\\
					Add $6 \times$ reg. 1 to $-2 \times$ reg. 2 & $\ket{x_0}$ & $\ket{x_0+2x_1+4x_2+8x_3}$ & $\ket{8x_0+8x_1+32x_2+32x_3}$ & $\ket{x_3}$\\
					Divide reg. 2 by 8 & $\ket{x_0}$ & $\ket{x_0+2x_1+4x_2+8x_3}$ & $\ket{x_0+x_1+4x_2+4x_3}$ & $\ket{x_3}$\\
					Add $-3 \times$ reg. 2 to $2 \times$ reg. 1 & $\ket{x_0}$ & $\ket{-x_0+x_1-4x_2+4x_3}$ & $\ket{x_0+x_1+4x_2+4x_3}$ & $\ket{x_3}$\\
					Add $3 \times$ reg. 0 to $4 \times$ reg. 1 & $\ket{x_0}$ & $\ket{-x_0+4x_1-16x_2+16x_3}$ & $\ket{x_0+x_1+4x_2+4x_3}$ & $\ket{x_3}$\\
					Add $12 \times$ reg. 3 to  reg. 2 & $\ket{x_0}$ & $\ket{-x_0+4x_1-16x_2+16x_3}$ & $\ket{x_0+x_1+4x_2+16x_3}$ & $\ket{x_3}$\\
					Add $-3 \times$ reg. 0 to $4 \times$ reg. 2 & $\ket{x_0}$ & $\ket{-x_0+4x_1-16x_2+16x_3}$ & $\ket{x_0+4x_1+16x_2+64x_3}$ & $\ket{x_3}$\\
					Add $48 \times$ reg. 3 to  reg. 1 & $\ket{x_0}$ & $\ket{-x_0+4x_1-16x_2+64x_3}$ & $\ket{x_0+4x_1+16x_2+64x_3}$ & $\ket{x_3}$\\
					\textbf{Product on regs. 1 and 2} & $\ket{x_0}$ & $\mathbf{\ket{-x_0+4x_1-16x_2+64x_3}}$ & $\mathbf{\ket{x_0+4x_1+16x_2+64x_3}}$ & $\ket{x_3}$\\
					Add  reg. 2 to  reg. 1 & $\ket{x_0}$ & $\ket{8x_1+128x_3}$ & $\ket{x_0+4x_1+16x_2+64x_3}$ & $\ket{x_3}$\\
					Divide reg. 1 by 8 & $\ket{x_0}$ & $\ket{x_1+16x_3}$ & $\ket{x_0+4x_1+16x_2+64x_3}$ & $\ket{x_3}$\\
					Add $-4 \times$ reg. 1 to  reg. 2 & $\ket{x_0}$ & $\ket{x_1+16x_3}$ & $\ket{x_0+16x_2}$ & $\ket{x_3}$\\
					Add $-1 \times$ reg. 0 to  reg. 2 & $\ket{x_0}$ & $\ket{x_1+16x_3}$ & $\ket{16x_2}$ & $\ket{x_3}$\\
					Divide reg. 2 by 16 & $\ket{x_0}$ & $\ket{x_1+16x_3}$ & $\ket{x_2}$ & $\ket{x_3}$\\
					Add $-16 \times$ reg. 3 to  reg. 1 & $\ket{x_0}$ & $\ket{x_1}$ & $\ket{x_2}$ & $\ket{x_3}$\\
					\hline
				\end{tabular}
		\end{center}
		\caption{
			\small
			\label{tab:k4_qq_ops}
			\textbf{$\mathbf{k=4}$ parallel sequence for \textsf{PhaseTripleProduct}.}
			In this table only the state of the $x$ sub-registers are shown; the same operations are applied to the $y$ and $z$ registers.
			Registers containing values upon which the algorithm is applied again recursively are highlighted in bold.
			The linear combinations used here for the products correspond to the evaluation points $w_\ell \in \{0, \infty, \pm 1, \pm 1/2, \pm 2, \pm 4 \}$.
			We note that it may be possible to parallelize the products further into three layers, rather than four, but the number of required additions would likely increase.
		}
	\end{table*}

\end{document}